\documentclass[11pt]{article} 
\usepackage[top=20mm,bottom=20mm, left=20mm, right=20mm]{geometry}
\usepackage{graphicx,amsmath,amssymb,url,enumerate,mathrsfs,subfig,epsfig,color,transparent,cite}
\usepackage{hyperref, makecell}
\usepackage{appendix, stmaryrd, theorem}
 \hypersetup{
     colorlinks=true,       
     linkcolor=red,          
     citecolor=blue,        
     filecolor=magenta,      
     urlcolor=cyan           
 }

\theorembodyfont{\rm}

\newtheorem{proposition}{Proposition}[section]
\newtheorem{theorem}{Theorem}[section]
\newtheorem{cor}{Corollary}[section]

\newtheorem{lemma}{Lemma}[section]
\newtheorem{rem}{\bf Remark}[section]

\newtheorem{ids}{\bf Identities}[section]

\def\div{\operatorname{div}}
\def\Div{\operatorname{Div}}

\def\curl{\operatorname{curl}}
\def\Curl{\operatorname{Curl}}
\def\sym{\operatorname{sym}}
\def\skw{\operatorname{skw}}
\def\tr{\operatorname{tr}}
\def\Lin{\operatorname{Lin}}
\def\Sym{\operatorname{Sym}}

\def\Skw{\operatorname{Skw}}\newenvironment{proof}[1][Proof]{\begin{trivlist}
\item[\hskip \labelsep {\bfseries #1}]}{\end{trivlist}}

\newcommand{\qed}{\nobreak \ifvmode \relax \else
      \ifdim\lastskip<1.5em \hskip-\lastskip
      \hskip1.5em plus0em minus0.5em \fi \nobreak
      \vrule height0.75em width0.5em depth0.25em\fi}

\title{Topological defects and metric anomalies as sources of incompatibility for piecewise smooth strain fields}
\author{Animesh Pandey and Anurag Gupta\thanks{ag@iitk.ac.in}
\\Department of Mechanical Engineering,\\ Indian Institute of Technology Kanpur, 208016, India.
}
\date{\today}
\begin{document}
\maketitle
\begin{abstract}
The incompatibility of linearized piecewise smooth strain field, arising out of volumetric and surface densities of topological defects and metric anomalies, is investigated. First, general forms of compatibility equations are derived for a piecewise smooth strain field, defined over a simply connected domain, with either a perfectly bonded or an imperfectly bonded interface. Several special cases are considered and discussed in the context of existing results in the literature. Next, defects, representing dislocations and disclinations, and metric anomalies, representing extra matter, interstitials, thermal, and growth strains, etc., are introduced in a unified framework which allows for incorporation of their bulk and surface densities, as well as for surface densities of defect dipoles.
Finally, strain incompatibility relations are derived both on the singular interface, and away from it, with sources in terms of defect and metric anomaly densities. With appropriate choice of constitutive equations, the incompatibility relations can be used to determine the state of internal stress within a body in response to the given prescription of defects and metric anomalies. 
\end{abstract}

{\small


}

\section{Introduction}

A central problem of micromechanics of defects in solids, in the context of linear elasticity, is to determine the internal stress field  for a given inhomogeneity field \cite{kroner68, kroner81, de1981view, ayan2}. The latter can be considered in terms of a density of topological defects, such as dislocations and disclinations, or metric anomalies, such as those engendered in problems of thermoelasticity, biological growth, interstitials, extra matter, etc. \cite{kroner81, ayan1}. The inhomogeneity field appears as a source in strain incompatibility relations, which when written in terms of stress, and combined with equilibrium equations and boundary conditions, yields the complete boundary value problem for the determination of internal stress field \cite{kroner81}. This classical problem of linear elasticity has been formulated, and solved, in the literature assuming the strain (and therefore stress) to be a smooth tensor field over the body. The defects densities have been also assumed, in general, to be smooth fields. The concern of the present paper is to generalize the problems of both strain compatibility and incompatibility with the consideration of \textit{piecewise smooth} strain and inhomogeneity fields. The bulk fields are therefore allowed to be discontinuous across a surface within the body. The developed framework, in addition, allows us to consider surface concentration of strain and inhomogeneity fields; it is also amenable to situations when these fields are concentrated on a curve within the body. 

In the strain compatibility problem, we seek necessary and sufficient conditions on a piecewise smooth symmetric tensor field (strain), defined over a simply connected domain, for there to a exist a piecewise smooth, but continuous (perfectly bonded interface), vector field (displacement) whose symmetric gradient is equal to the tensor field. The conditions consist of the well known compatibility condition on the strain field, away from the singular interface, and the jump conditions on strain and its gradients across the interface. The conditions are also sought for the case when the displacement field is no longer required to be continuous (imperfectly bonded interface). This, however, necessarily requires us to consider a concentration of surface strain field on the interface. The general forms of compatibility conditions, obtained in both the cases, are novel to the best of our knowledge. They are reduced to several specific situations discussed previously in the literature. We recover the  interfacial jump conditions obtained by Markenscoff \cite{markenscoff1996note} and Wheeler and Luo \cite{wheeler1999conditions}. Whereas the former work was restricted to plane strain, the latter was concerned only with perfectly bonded interfaces and expressing the jump conditions in terms of strain components with respect to a specific curvilinear basis. We also use our framework to obtain the compatibility conditions on smooth strain fields over a domain, on a part of whose boundary displacements are specified, as discussed recently by Ciarlet and Mardare \cite{ciarlet2014intrinsic}.

A strain field is termed incompatible if it does not satisfy the compatibility conditions. There can then no longer exist a displacement field whose symmetric gradient will be equal to the strain field, and hence the strain can not correspond to a physical deformation. The loss of compatibility is attributed to inhomogeneity fields in terms of defects and metric anomalies \cite{kroner81, de1981view}. In our work we consider piecewise smooth bulk densities, and smooth surface densities (or surface concentrations), of dislocations and disclinations. We also allow for smooth surface densities of defect dipoles. In addition we consider piecewise smooth bulk density, and smooth surface density, of metric anomalies. Beginning with writing these densities in terms of kinematical quantities, such as strain and bend-twist field, we first obtain the conservations laws they should necessarily satisfy. Our formulation is then led towards relating incompatibility of the strain field with densities of defects and metric anomalies. The strain incompatibility relations thus derived, with weaker regularity in the strain and inhomogeneity fields, as compared to the existing literature, are the central results of this paper. The incompatibility itself is described in terms of a piecewise smooth bulk field and smooth surface concentrations.

A brief outline of the paper is as follows. In Section \ref{prem}, the required mathematical infrastructure is developed. Several elements of the theory of distributions, which forms the backbone of our work, are discussed. The results, already available in the literature, are given without proof but otherwise self-contained proofs are provided within the section and in the appendix.
The strain compatibility problem, first for a perfectly bonded and then for an imperfectly bonded interface, is addressed in Section \ref{cdsf}. Several remarks are provided in order to connect our results with the existing literature as well as to gain further insights. In Section \ref{inhomoincomp}, the central problem of strain incompatibility arising in response to the given inhomogeneity fields is formulated. Various aspects of the theory are simplified and discussed in the context of defect conservation laws, dislocation loops, plane strain simplification, and nilpotent defect densities. The paper concludes in Section \ref{conc}.

\section{Mathematical Preliminaries}
\label{prem}
\subsection{Notation}
Let $\Omega \subset \mathbb{R}^3$ be a bounded, connected, open set, with a smooth boundary $\partial \Omega$. For two sets $A$ and $B$, $A-B$ denotes the difference between the sets, whereas $\emptyset$ represents the empty set. The Greek indices range over $\{ 1,2\}$ and the Latin indices range over $\{ 1,2,3\}$. Let $\{ \boldsymbol{e}_1, \boldsymbol{e}_2, \boldsymbol{e}_3\}$ be a fixed orthonormal right-handed basis in $\mathbb{R}^3$. For $\boldsymbol{u},\boldsymbol{v} \in \mathbb{R}^3$, the inner product is given by $\langle \boldsymbol{u},\boldsymbol{v} \rangle = u_i v_i$, where $u_i=\langle \boldsymbol{u},\boldsymbol{e}_i\rangle$, etc; here, and elsewhere, summation is implied over repeated indices, unless stated otherwise. The cross product $\boldsymbol{u} \times \boldsymbol{v} \in \mathbb{R}^3$ is such that $\left( \boldsymbol{u} \times \boldsymbol{v} \right)_i=\epsilon_{ilk} u_l v_k$, where $\epsilon_{ilk}$ is the alternating symbol.  We use $\Lin$ to represent the space of  second order tensors (or, in other words, the linear transformations from  $\mathbb{R}^3$ to itself) and $\Sym$, $\Skw$ the space of symmetric and skew symmetric second order tensors, respectively. The identity tensor in $\Lin$ is denoted by $\boldsymbol{I}$. The dyadic product $\boldsymbol{u}\otimes \boldsymbol{v} \in \Lin$ is defined such that $(\boldsymbol{u}\otimes \boldsymbol{v})\boldsymbol{w}=\langle\boldsymbol{v},\boldsymbol{w} \rangle \boldsymbol{u}$, where $\boldsymbol{w}\in \mathbb{R}^3$. For $\boldsymbol{a} \in \Lin $, $\boldsymbol{a}^T$, $\sym(\boldsymbol{a})$, and $\skw(\boldsymbol{a})$ represent the transpose, the symmetric part, and the skew part of $\boldsymbol{a}$, respectively. The axial vector of $\boldsymbol{b} \in \Skw$ is $ax(\boldsymbol{b}) \in  \mathbb{R}^3$ such that, for any $\boldsymbol{v}\in \mathbb{R}^3$, $\boldsymbol{b}\boldsymbol{v}=ax(\boldsymbol{b})\times \boldsymbol{v}$. For $\boldsymbol{a},\boldsymbol{c} \in \Lin $, the inner product is given by $\langle \boldsymbol{a},\boldsymbol{c} \rangle = a_{ij} c_{ij}$ with $a_{ij}=\langle \boldsymbol{a},\boldsymbol{e}_i \otimes \boldsymbol{e}_j \rangle$, etc. The trace of $\boldsymbol{a} \in \Lin$ is defined as $\tr(\boldsymbol{a})=\langle\boldsymbol{a},\boldsymbol{I} \rangle$. For $\boldsymbol{a} \in \Lin$ and $\boldsymbol{v} \in \mathbb{R}^3$, we define $\boldsymbol{a} \times \boldsymbol{v} \in \Lin$ such that
 $(\boldsymbol{a} \times \boldsymbol{v})_{ji} = -\epsilon_{ilk} a _{jk} v_l$. For $\boldsymbol{a} \in \Lin$ and $\boldsymbol{b} \in \Lin$, we define $\boldsymbol{a}\times \boldsymbol{b}$, a linear map from $\mathbb{R}^{3}$ to $\Lin$,  such that, for any $\boldsymbol{v}\in \mathbb{R}^3$, $(\boldsymbol{a}\times \boldsymbol{b})\boldsymbol{v}=(\boldsymbol{a}\times (\boldsymbol{b}^T \boldsymbol{v}))^T$. 
 
Let $S \subset \Omega$ be a regular oriented surface with unit normal $\boldsymbol{n}$ and boundary $\partial S$. If $\partial S - \partial \Omega = \emptyset$, then $S$ is either a closed surface or its boundary is completely contained within the boundary of $\Omega$. In either case, $S$ will divide $\Omega$ into mutually exclusive open sets $\Omega^+$ and $\Omega^-$ such that $\partial \Omega^+ \cap \partial \Omega^- = S$ and $\Omega^+ \cup S \cup \Omega^-=\Omega$. The set $\Omega^-$ is the one into which $\boldsymbol{n}$ points.
 
We use $C^0 (\Omega)$, $C^{\infty}(\Omega)$ and $C^{r}(\Omega)$ ($r$ is a positive integer), to represent spaces of continuous, smooth, and $r$-times differentiable functions on $\Omega$, respectively. The spaces of vector valued and tensor valued smooth functions on $\Omega$ are represented by $C^{\infty}(\Omega, \mathbb{R}^3)$ and $C^{\infty}(\Omega, \Lin)$, respectively. Similar notations are used for functions defined over surface $S$. For a function $f$ on $\Omega$ and a subset $\omega \subset \Omega$, $f|_\omega$ is the restriction of $f$ to the subset $\omega$.

\subsection{Distributions}

Let $\mathcal{D}(\Omega)$ be the space of compactly supported smooth functions on $\Omega$. The dual space of $\mathcal{D}(\Omega)$ is the space of distributions, $\mathcal{D}'(\Omega)$. Any distribution $T\in \mathcal{D}'(\Omega)$ defines a linear functional  $T : \mathcal{D}(\Omega)\to \mathbb{R}$ which is continuous for an appropriately defined topology on  $\mathcal{D}(\Omega)$ \cite[Chapter~1]{kesavantopics}.\footnote{A sequence of smooth functions $\phi_m \in \mathcal{D}(\Omega)$ converges to 0 if $\phi_m$, for all $m$, are supported in a fixed compact support and $\phi_m$ and its derivatives to every order converge uniformly to 0.  A functional $T$ is continuous if, for any sequence of smooth functions $\phi_m \in \mathcal{D}(\Omega)$ converging to 0, $T(\phi_m)$ converges to 0.}
For the purpose of this article, we will be interested in certain types of distributions contained in $\mathcal{D}'(\Omega)$. 
For $\phi \in \mathcal{D}(\Omega)$, we say that a distribution $B \in \mathcal{B}(\Omega) \subset \mathcal{D}' (\Omega)$ if it is of the form
\begin{equation}
B(\phi)=\int_{\Omega} b \phi dv, \label{distB}
\end{equation}
where $b$ is a piecewise smooth function, possibly discontinuous across $S$ with $\partial S - \partial \Omega = \emptyset$, and $dv$ is the volume measure on $\Omega$. The discontinuity in $b$ is assumed to be a smooth function on $S$. For $x \in S$, $\llbracket b \rrbracket \left(x\right) = b^+ \left(x\right) - b^- \left(x\right)$, where $b^{\pm}\left(x\right)$ are limiting values of $b$ at $x$ on $S$ from $\Omega^{\pm}$, represents the discontinuity in $b$.
We say that a distribution $C \in \mathcal{C}\left(\Omega\right) \subset \mathcal{D}' (\Omega)$ if it is of the form,
\begin{equation}
C\left(\phi\right)=\int_{S} c \phi da, \label{distC}
\end{equation}
where $c$, the surface density of $C$, is assumed to be a smooth function on $S$ and $da$ is the area measure on the surface. 
We say that a distribution $F \in \mathcal{F}\left(\Omega\right) \subset \mathcal{D}' (\Omega)$ if it is of the form
\begin{equation}
F\left(\phi\right)=\int_{S} f \frac{\partial \phi}{\partial {n}} da, \label{distF}
\end{equation}
where $f$ is assumed to be a smooth function on $S$ and ${\partial}/{\partial {n}}$ represents the partial derivative along $\boldsymbol{n}$, i.e., ${\partial \phi}/{\partial n}= \langle \nabla \phi,\boldsymbol{n}\rangle$ (here $\nabla \phi$ denotes the gradient of $\phi$).
We say that a distribution $H \in \mathcal{H}\left(\Omega\right) \subset \mathcal{D}' (\Omega)$ if it is of the form
\begin{equation}
H\left(\phi\right)=\int_{L} h \phi dl, \label{distH}
\end{equation}
where $h$ is assumed to be a smooth function on a smooth oriented curve $L\subset \Omega$ and $dl$ is the length measure on $L$. 
That the above defined functionals are indeed distributions can be verified by first noting that all of them are linear functionals on $\mathcal{D}(\Omega)$. We now establish their continuity on $\mathcal{D}(\Omega)$. From $\phi_m \in \mathcal{D}(\Omega)$ converging to $0$ it is implied that for $\epsilon >0$ there exist positive integers $m_0$, $m_1$ such that $|\phi_m (x)| < \epsilon $ for $m>m_0$ and $|{\partial \phi_m (x)}/{\partial n}| < \epsilon $ for  $m>m_1$. For $B(\phi)=\int b\phi dv$, $|B(\phi_m)| \leq \sup(|b|) V \epsilon$, where $V$ is the volume of $\Omega$. Hence, $B(\phi_m)$ converges to 0. Similar arguments hold for $C\left(\phi\right)$, $F\left(\phi\right)$, and $H\left(\phi\right)$. 

We use $\mathcal{D}(\Omega,\mathbb{R}^3)$ to denote be the space of compactly supported vector valued smooth functions on $\Omega$. The corresponding dual space is the space of vector valued distributions, $\mathcal{D}'(\Omega,\mathbb{R}^3)$. For $\boldsymbol{T} \in \mathcal{D}'(\Omega,\mathbb{R}^3)$, with each component $T_i \in \mathcal{D}'(\Omega)$, and $\boldsymbol{\phi} \in \mathcal{D}(\Omega,\mathbb{R}^3)$, we define $\boldsymbol{T}(\boldsymbol{\phi})= T_{i} (\phi_{i})$ (summation is implied over repeated indices).  Analogously, the space of compactly supported tensor valued function on $\Omega$ and its dual are represented by $\mathcal{D}(\Omega,\Lin)$ and $\mathcal{D}'(\Omega,\Lin)$, respectively.  For $\boldsymbol{T} \in \mathcal{D}'(\Omega,\Lin)$, with each component $T_{ij} \in \mathcal{D}'(\Omega)$, and $\boldsymbol{\phi} \in \mathcal{D}(\Omega,\Lin)$, we define $\boldsymbol{T}(\boldsymbol{\phi})=T_{ij} (\phi_{ij})$.

\subsection{Derivatives of Distributions}

 The \textit{partial derivative} of a distribution $T\in \mathcal{D}'(\Omega)$ is a distribution ${\partial_i T} \in \mathcal{D}'(\Omega)$ defined as
 \begin{equation}
{\partial_i T} \left(\phi\right)= - T\left(\frac{\partial \phi}{\partial x_i}\right)
 \end{equation}
for all $\phi \in \mathcal{D}(\Omega)$ with $\boldsymbol{x} \in \Omega$.\footnote{Any locally integrable function $f$ can be associated with a distribution $T_f \in \mathcal{D}'(\Omega)$ such that, for all $\phi \in \mathcal{D}(\Omega)$, 
\begin{equation}
T_f (\phi)= \int_{\Omega} f \phi dv.
\end{equation}
 For a differentiable function $f \in C^1 (\Omega)$, 
\begin{equation}
{\partial_i T_f} \left(\phi\right)=-\int_{\Omega} f \frac{\partial \phi}{\partial x_i} dv = \int_{\Omega}  \frac{\partial f}{\partial x_i} \phi dv.
\end{equation} 
Hence, ${\partial_i T_f} = T_{\frac{\partial f}{\partial x_i}}$. The definition of partial derivative for distributions therefore generalises the notion of partial derivative for differentiable functions.}
The higher order derivatives can be consequently defined. For instance, the second order partial derivative of $T$ is a distribution ${\partial^2_{ij} T} \in \mathcal{D}'(\Omega)$ given by 
\begin{equation}
{\partial^2_{ij} T } \left(\phi\right)= T \left(\frac{\partial^2 \phi }{\partial x_i \partial x_j}\right), \label{diff_sym}
\end{equation}
which implies ${\partial^2_{ji} T}={\partial^2_{ij} T}$. The \textit{gradient} of a scalar distribution $T \in \mathcal{D}'(\Omega)$ is a vector valued distribution $\nabla T \in \mathcal{D}'(\Omega, \mathbb{R}^3)$ such that $(\nabla T)_i = {\partial_i T}$. The gradient of a vector valued distribution $\boldsymbol{T}\in \mathcal{D}'(\Omega,\mathbb{R}^3)$ is a tensor valued distribution $\nabla \boldsymbol{T} \in \mathcal{D}'(\Omega,\Lin)$ such that $(\nabla \boldsymbol{T})_{ij}={\partial_j T_i}$.  
The \textit{divergence} of a vector valued distribution $\boldsymbol{T}\in \mathcal{D}'(\Omega,\mathbb{R}^3)$ is a scalar valued distribution $\Div \boldsymbol{T} \in \mathcal{D}'(\Omega)$ such that $\Div \boldsymbol{T}={\partial_i T_i}$.  The divergence of a tensor valued distribution $\boldsymbol{T}\in \mathcal{D}'(\Omega,\Lin)$ is a vector valued distribution $\Div \boldsymbol{T} \in \mathcal{D}'(\Omega,\mathbb{R}^3)$ such that $(\Div \boldsymbol{T})_i={\partial_j T_{ij}}$. 
The \textit{curl} of a vector valued distribution $\boldsymbol{T} \in D' (\Omega,\mathbb{R}^3)$ is a vector valued distribution $\Curl \boldsymbol{T} \in \mathcal{D}'(\Omega,\mathbb{R}^3)$ such that $(\Curl \boldsymbol{T})_i = \epsilon_{ijk} {\partial_j T_k}$.
The \textit{curl} of a tensor valued distribution $\boldsymbol{T} \in D' (\Omega,\Lin)$ is a tensor valued distribution $\Curl \boldsymbol{T} \in \mathcal{D}'(\Omega,\Lin)$ such that $(\Curl \boldsymbol{T})_{ij} = \epsilon_{ilk} {\partial_l T_{jk}}$. In particular, for $\boldsymbol{T} \in D' (\Omega,\Lin)$, we have a tensor valued distribution $\Curl \Curl \boldsymbol{T} \in \mathcal{D}'(\Omega,\Lin)$ such that $(\Curl \Curl \boldsymbol{T})_{ij} = \epsilon_{ilk} \epsilon_{jmn} {\partial^2_{lm} T_{kn}}$.

\subsection{Derivatives of Smooth Fields} 

The gradients of a smooth scalar field ${v} \in C^{\infty}(\Omega)$ and a smooth vector field $\boldsymbol{v} \in C^{\infty}(\Omega,\mathbb{R}^3)$ are denoted by $\nabla {v} \in C^{\infty}(\Omega,\mathbb{R}^3)$ and  $\nabla \boldsymbol{v} \in C^{\infty}(\Omega,\Lin)$, respectively.
The divergence of $\boldsymbol{v}$ is a smooth scalar field defined as $\div \boldsymbol{v} = \tr(\nabla \boldsymbol{v})$. The divergence of a smooth tensor field $\boldsymbol{a} \in C^{\infty}(\Omega,\Lin)$ is a smooth vector field $\div \boldsymbol{a}$ defined by $\langle \div \boldsymbol{a},\boldsymbol{d} \rangle=\div(\boldsymbol{a}^T \boldsymbol{d})$, for any fixed $\boldsymbol{d}\in \mathbb{R}^3$.
The curl of $\boldsymbol{v}$ is a smooth vector field $\curl \boldsymbol{v}$ defined as $\langle \curl \boldsymbol{v},\boldsymbol{d} \rangle=\div(\boldsymbol{v}\times \boldsymbol{d})$, for any fixed $\boldsymbol{d}\in \mathbb{R}^3$. The curl of $\boldsymbol{a}$ is a smooth tensor field $\curl \boldsymbol{a}$ defined as $(\curl \boldsymbol{a})\boldsymbol{d}=\curl(\boldsymbol{a}^T \boldsymbol{d})$, for any fixed $\boldsymbol{d}\in \mathbb{R}^3$.
The gradient of a scalar distribution $T \in \mathcal{D}'(\Omega)$ can be therefore be equivalently defined as $\nabla T (\boldsymbol{\phi})=-T(\div \boldsymbol{\phi})$, for all $\boldsymbol{\phi} \in \mathcal{D}(\Omega,\mathbb{R}^3)$. Similarly, the divergence of a vector valued distribution $\boldsymbol{T} \in \mathcal{D}'(\Omega,\mathbb{R}^3)$ can be equivalently defined as $\Div \boldsymbol{T} ({\phi})=-\boldsymbol{T}(\nabla {\phi})$, for all ${\phi} \in \mathcal{D}(\Omega)$.
Furthermore, we can define the curl of a tensor valued distribution $\boldsymbol{T} \in \mathcal{D}'(\Omega,\Lin)$ as 
$(\Curl \boldsymbol{T})(\boldsymbol{\phi}^T) = \boldsymbol{T}\left((\curl\boldsymbol{\phi})^T\right)$, for all $\boldsymbol{\phi} \in \mathcal{D}(\Omega,\Lin)$.

The surface gradient of a smooth field $v \in  C^{\infty}(S)$, with a smooth extension $\overline{v} \in  C^{\infty}(\Omega)$, i.e., $\overline{v} = v$ on $S$, is a smooth vector field $\nabla_S {v} \in C^{\infty}(S,\mathbb{R}^3)$ obtained by projecting $\nabla \overline{v}$ onto the tangent plane of the surface. The surface gradient of a smooth vector field $\boldsymbol{v} \in C^{\infty}(S,\mathbb{R}^3)$ is a smooth tensor field $\nabla_S \boldsymbol{v}  \in C^{\infty}(S,\Lin)$ such that $\nabla_S \boldsymbol{v}=\nabla \overline{\boldsymbol{v}} (\boldsymbol{I}-\boldsymbol{n}\otimes \boldsymbol{n})$, where $\overline{\boldsymbol{v}} \in C^{\infty}(\Omega,\mathbb{R}^3)$ is a smooth extension of $\boldsymbol{v}$ (i.e., $\overline{\boldsymbol{v}} = {\boldsymbol{v}}$ on $S$). 
The surface divergence of $\boldsymbol{v} \in C^{\infty}(S,\mathbb{R}^3)$ is a smooth scalar field $\div_S \boldsymbol{v} \in  C^{\infty}(S)$ defined as $\div_S \boldsymbol{v}=\tr(\nabla_S \boldsymbol{v})$. In terms of the extension $\overline{\boldsymbol{v}}$, it is given by $\div_S \boldsymbol{v}=\div \overline{\boldsymbol{v}} - \left\langle\left(\nabla \overline{\boldsymbol{v}}\right) \boldsymbol{n},\boldsymbol{n}\right\rangle$. In particular, the scalar field  $\kappa = -\div_S \boldsymbol{n}$ is twice the mean curvature of surface $S$. 
The surface divergence of a tensor field $\boldsymbol{a} \in C^{\infty}(\Omega,\Lin) $ is a vector field $\div_S \boldsymbol{a} \in C^{\infty}(S,\mathbb{R}^3)$ defined by $\langle\div_S \boldsymbol{a},\boldsymbol{d}\rangle=\div_S (\boldsymbol{a}^T \boldsymbol{d})$. In terms of a smooth extension $\overline{\boldsymbol{a}}  \in C^{\infty}(\Omega,\Lin)$, it is given by
$\div_S \boldsymbol{a}=\div \overline{\boldsymbol{a}}-\left( \left( \nabla \overline{\boldsymbol{a}} \right) \boldsymbol{n}\right)\boldsymbol{n}$. Finally, if $\boldsymbol{a}$ is a linear map from $\mathbb{R}^3$ to $\Lin$ (third order tensor), the surface divergence $\div_S \boldsymbol{a} \in \Lin$ is given by 
$(\div_S \boldsymbol{a})\boldsymbol{d}=\div_S (\boldsymbol{a}\boldsymbol{d})$.

Motivated by the definition of curl of vector fields on $\Omega$, we introduce, for $\boldsymbol{v} \in C^\infty(S,\mathbb{R}^3)$,   a vector valued smooth field $\curl_S \boldsymbol{v} \in C^\infty(S,\mathbb{R}^3)$ such that, for any fixed $\boldsymbol{d} \in \mathbb{R}^3$, $
\langle \curl_S \boldsymbol{v}, \boldsymbol{d}\rangle = \div_S \left(\boldsymbol{v} \times \boldsymbol{d}\right)$.
Analogous to its bulk counterpart, $\curl_S \boldsymbol{v}$ gives the axial vector of $(\nabla_S \boldsymbol{v} - (\nabla_S \boldsymbol{v})^T)$. If $\boldsymbol{v}$ has no tangential component, i.e., $\boldsymbol{v}=v \boldsymbol{n}$ with $v \in  C^{\infty}(S)$, then we obtain $2 \skw(\nabla \boldsymbol{v})=\nabla_S v \otimes \boldsymbol{n} - \boldsymbol{n} \otimes \nabla_S v$.  On the other hand, if we consider $\boldsymbol{v}$ to be tangential and $S$ to be planar, i.e., $\langle\boldsymbol{v},\boldsymbol{n} \rangle =0$ and $\nabla_S \boldsymbol{n} = \boldsymbol{0}$, then we have $\curl_S \boldsymbol{v} =\langle\curl \overline{\boldsymbol{v}},\boldsymbol{n}\rangle\boldsymbol{n}$, where $\overline{\boldsymbol{v}}$ is a smooth extension of $\boldsymbol{v}$ over $\Omega$. More generally, the following relationship holds:
\begin{equation}
\curl_S \boldsymbol{v}= \left(\frac{\partial \overline{\boldsymbol{v}}}{\partial {n}} \times \boldsymbol{n}\right)  + \curl\left(\overline{\boldsymbol{v}}\right) ~\text{on}~S. \label{scurlv}
\end{equation}
For  $\boldsymbol{a} \in C^\infty(S,\Lin)$, we introduce a tensor valued smooth field $\curl_S \boldsymbol{a} \in C^\infty(S,\Lin)$ such that, for any fixed $\boldsymbol{d} \in \mathbb{R}^3$, $\left(\curl_S \boldsymbol{a}\right)^T \boldsymbol{d} = \div_S \left(\boldsymbol{a} \times \boldsymbol{d}\right)$. 
In terms of a smooth extension $\overline{\boldsymbol{a}} \in C^\infty\left(\Omega,\Lin\right)$ of $\boldsymbol{a}$,  such that $\overline{\boldsymbol{a}} = {\boldsymbol{a}}$ on $S$,  
\begin{equation}
\curl_S \boldsymbol{a}= \left(\frac{\partial \overline{\boldsymbol{a}}}{\partial {n}} \times \boldsymbol{n}\right)^T  + \curl\left(\overline{\boldsymbol{a}}\right) ~\text{on}~S.
\end{equation}
Indeed, for fixed vectors $\boldsymbol{d} \in \mathbb{R}^3$ and $\boldsymbol{f} \in \mathbb{R}^3$, we can use the identity $\left(\boldsymbol{a}\times \boldsymbol{d}\right)^T \boldsymbol{f} = \left(\boldsymbol{a}^T \boldsymbol{f} \times \boldsymbol{d} \right)$ to obtain
\begin{equation}
\left\langle \left(\left(\frac{\partial \overline{\boldsymbol{a}}}{\partial {n}} \times \boldsymbol{n}\right)^T  + \curl\left(\overline{\boldsymbol{a}}\right)\right)\boldsymbol{f} , \boldsymbol{d} \right\rangle = \left\langle \left(\frac{\partial \overline{\boldsymbol{a}}}{\partial {n}} \times \boldsymbol{n}\right)\boldsymbol{d},\boldsymbol{f} \right\rangle + \div \left(\left(\overline{\boldsymbol{a}}^T \boldsymbol{f}\right)\times \boldsymbol{d}\right).
\end{equation}
Consequent to writing the divergence term above in terms of a surface divergence, and proceeding with straightforward manipulations, we obtain the desired result.
Equation \eqref{scurlv} can be established along similar lines. It is clear that these relationships are independent of the choice of an extension.

Given a smooth oriented curve $L \subset \Omega$, with tangent $\boldsymbol{t} \in C^\infty (L,\mathbb{R}^3)$, consider a surface $S(x_0)$ passing through point $x_0 \in L$ such that $\boldsymbol{t}(x_0)$ is the normal to $S(x_0)$ at $x_0$. For a smooth bulk vector field $\boldsymbol{v} \in C^\infty (\Omega,\mathbb{R}^3)$, we define a vector valued smooth field $\curl_{t} \boldsymbol{v} \in C^\infty (L,\mathbb{R}^3)$ such that, at any $x_0 \in L$,
$\curl_{t} \boldsymbol{v}=\curl_{S(x_0)} (\boldsymbol{v}|_{S(x_0)})$, which is equal to $(({\partial \boldsymbol{v}}/{\partial {t}}) \times \boldsymbol{t} )  + \curl {\boldsymbol{v}}$ by Equation \eqref{scurlv}, where $\partial / \partial t$ is the derivative along $\boldsymbol{t}$.
It is immediate that this definition is independent of the choice of the surface $S(x_0)$ 
as long as the normal to $S(x_0)$ at $x_0$ is $\boldsymbol{t}$.

\subsection{Useful Identities}
\label{ui}

In this section we collect several identities which relate derivatives of distributions to derivatives of smooth functions. These identities will be central to the rest of our work. The proofs of these identities are collected in Appendix \ref{appid}.

\begin{ids}
\label{GradientLemma}
(Gradient of distributions) For $\boldsymbol{\psi} \in \mathcal{D}(\Omega,\mathbb{R}^3)$,

\noindent (a) If ${B} \in \mathcal{B}(\Omega)$, as defined in Equation \eqref{distB}, then
\begin{equation}
\nabla B (\boldsymbol{\psi})= \int_\Omega \langle \nabla b , \boldsymbol{\psi}\rangle dv - \int_S \left\langle \llbracket b \rrbracket \boldsymbol{n} , \boldsymbol{\psi} \right\rangle da. \label{gida}
\end{equation}

\noindent (b) If ${C} \in \mathcal{C}(\Omega)$, as defined in Equation \eqref{distC}, then
\begin{equation}
 \label{gidb}
\nabla C (\boldsymbol{\psi})= -\int_{\partial S - \partial \Omega} \langle c\boldsymbol{\nu},\boldsymbol{\psi} \rangle dl + \int_S \left\langle \left( \nabla_S c + \kappa c \boldsymbol{n} \right), \boldsymbol{\psi} \right\rangle da - \int_S  \left\langle c \boldsymbol{n} , \frac{\partial \boldsymbol{\psi}}{\partial n} \right\rangle da,
\end{equation}
where $\boldsymbol{\nu}$ is the in plane normal to $\partial S - \partial \Omega$.

\noindent (c) If ${F} \in \mathcal{F}(\Omega)$, as defined in Equation \eqref{distF}, then
\begin{equation}
 \label{gidc}
\begin{split}
\nabla F (\boldsymbol{\psi})= -\int_{\partial S-\partial \Omega} \left\langle f\boldsymbol{\nu},\frac{\partial \boldsymbol{\psi}}{\partial n} \right\rangle dl + \int_{\partial S-\partial \Omega} \left\langle f (\nabla_S \boldsymbol{n}) \boldsymbol{\nu},\boldsymbol{\psi}\right\rangle dl +\int_S \left\langle (\nabla_S f+ \kappa f \boldsymbol{n}), \frac{\partial \boldsymbol{\psi}}{\partial n} \right\rangle da \\- \int_S  \left\langle \div_S (f\nabla_S \boldsymbol{n}),  \boldsymbol{\psi} \right\rangle da  - \int_S  \left\langle f\boldsymbol{n}, \nabla(\nabla \boldsymbol{\psi}) \boldsymbol{n}\otimes \boldsymbol{n}\right\rangle) da.
\end{split}
\end{equation}

\noindent (d) If $H \in \mathcal{H}(\Omega)$, as defined in Equation \eqref{distH}, then
\begin{equation}
 \label{gidd}
\nabla H (\boldsymbol{\psi})= -\int_L \left( h \langle \nabla \boldsymbol{\psi},(\boldsymbol{I}-\boldsymbol{t}\otimes \boldsymbol{t})\rangle - \left\langle  \frac{\partial (h \boldsymbol{t})}{\partial t} , \boldsymbol{\psi}\right\rangle \right) dl - \langle h\boldsymbol{t},\boldsymbol{\psi} \rangle|_{\partial L - \partial \Omega},
\end{equation}
where $\boldsymbol{t}$ is the unit tangent along $L$. The last term above evaluates the function at the end points of $L$ (excluding those which lie on $\partial \Omega$) and should appropriately take into consideration the orientation of the curve at the evaluation point.
\end{ids}

The following two sets of identities are used to calculate divergence and curl of vector valued distributions $\boldsymbol{B} \in \mathcal{B}(\Omega,\mathbb{R}^3)$, $\boldsymbol{C} \in  \mathcal{C}(\Omega,\mathbb{R}^3)$, $\boldsymbol{F} \in \mathcal{F}(\Omega,\mathbb{R}^3)$, and $\boldsymbol{H} \in \mathcal{H}(\Omega,\mathbb{R}^3)$ such that, for $\boldsymbol{\phi} \in \mathcal{D}(\Omega,\mathbb{R}^3)$, 
\begin{equation}
\boldsymbol{B}(\boldsymbol{\phi})= \int_{\Omega} \langle \boldsymbol{b},\boldsymbol{\phi} \rangle dv,~\boldsymbol{C}(\boldsymbol{\phi})= \int_S \langle \boldsymbol{c},\boldsymbol{\phi} \rangle da,~\boldsymbol{F}(\boldsymbol{\phi})= \int_S \left\langle \boldsymbol{f},\frac{\partial \boldsymbol{\phi}}{\partial n}  \right\rangle da,~\text{and}~\boldsymbol{H}(\boldsymbol{\phi})=\int_L \langle \boldsymbol{h},\boldsymbol{\phi} \rangle dl, \label{vectdist}
\end{equation}
where $\boldsymbol{b}$ is a piecewise smooth vector valued function on $\Omega$, possibly discontinuous across $S$ with $\partial S - \partial \Omega = \emptyset$, $\boldsymbol{c}$ and $\boldsymbol{f}$ are smooth vector valued functions on $S$, and $\boldsymbol{h}$ is a smooth vector valued function on $L$.
The divergence and curl of a tensor valued distribution $\boldsymbol{A}\in \mathcal{D}'(\Omega,\Lin)$ can be obtained from the results for vector valued distributions using the identities $\langle\Div \boldsymbol{A}, \boldsymbol{d}\rangle=\Div(\boldsymbol{A}^T \boldsymbol{d})$ and $(\Curl \boldsymbol{A}) \boldsymbol{d}=\Curl(\boldsymbol{A}^T \boldsymbol{d})$ for any fixed vector $\boldsymbol{d} \in \mathbb{R}^3$.

\begin{ids}
\label{DivergenceLemma} (Divergence of distributions) 
For ${\psi} \in \mathcal{D}(\Omega)$,

\noindent (a) If $\boldsymbol{B} \in \mathcal{B}(\Omega,\mathbb{R}^3)$ then
\begin{equation}
\begin{split}
\Div \boldsymbol{B} \left({\psi}\right)= \int_{\Omega}  (\div \boldsymbol{b}) {\psi}  dv - \int_{S} \left\langle \llbracket \boldsymbol{b} \rrbracket , \boldsymbol{n} \right\rangle {\psi} da. 
\end{split} \label{DivB}
\end{equation}

\noindent (b) If $\boldsymbol{C} \in  \mathcal{C}(\Omega,\mathbb{R}^3)$ then
\begin{equation}
\label{DivC}
\Div \boldsymbol{C} \left({\psi}\right)=  \int_{S}   \left(\div_S \boldsymbol{c} + \kappa \langle\boldsymbol{c} ,\boldsymbol{n}\rangle \right){\psi} da -\int_S \left\langle \boldsymbol{c}, \boldsymbol{n}   \right\rangle \frac{\partial {\psi}}{\partial {n}}da - \int_{\partial S - \partial \Omega}  \left\langle \boldsymbol{c}, \boldsymbol{\nu}  \right\rangle{\psi} dl.
\end{equation}

\noindent (c) If  $\boldsymbol{F} \in \mathcal{F}(\Omega,\mathbb{R}^3)$ then
\begin{equation}
\label{DivF}
\begin{split}
\Div \boldsymbol{F}({\psi})=\int_{\partial S - \partial \Omega} \left\langle (\nabla_S \boldsymbol{n}) \boldsymbol{f},  \boldsymbol{\nu} \right\rangle {\psi}  dl-\int_{\partial S - \partial \Omega} \langle \boldsymbol{f},\boldsymbol{\nu} \rangle  \frac{\partial {\psi}}{\partial n} dl - \int_S \div_S \left( (\nabla_S \boldsymbol{n}) \boldsymbol{f} \right) {\psi} da \\+ \int_S \left(\div_S \boldsymbol{f} + \kappa \langle\boldsymbol{f} ,\boldsymbol{n}\rangle \right) \frac{\partial {\psi}}{\partial n}  da - \int_S \langle \boldsymbol{f},\boldsymbol{n} \rangle \langle \nabla(\nabla {\psi}), \boldsymbol{n} \otimes \boldsymbol{n} \rangle da.
\end{split}
\end{equation}

\noindent (d) If  $\boldsymbol{H} \in \mathcal{H}(\Omega,\mathbb{R}^3)$ then
\begin{equation}
\label{DivH}
\Div \boldsymbol{H} ({\psi})=-\int_L \left\langle (\boldsymbol{I}-\boldsymbol{t}\otimes \boldsymbol{t}) \boldsymbol{h}, (\boldsymbol{I}-\boldsymbol{t}\otimes \boldsymbol{t}) \nabla {\psi}   \right\rangle dl + \int_L \frac{\partial (\langle\boldsymbol{h},\boldsymbol{t}\rangle)}{\partial {t}} {\psi}  dl - (\langle \boldsymbol{h},\boldsymbol{t} \rangle{\psi}) |_{\partial L -\partial \Omega}.
\end{equation}
\end{ids}

\begin{ids}
\label{CurlLemma} (Curl of distributions) 
For $\boldsymbol{\phi} \in \mathcal{D}(\Omega,\mathbb{R}^3)$,

\noindent (a) If $\boldsymbol{B} \in  \mathcal{B}(\Omega,\mathbb{R}^3)$ then
\begin{equation}
\Curl \boldsymbol{B} \left(\boldsymbol{\phi}\right)=\int_{\Omega} \left\langle \curl \boldsymbol{b}, \boldsymbol{\phi} \right\rangle dv +\int_S \left\langle \left(\llbracket \boldsymbol{b} \rrbracket \times \boldsymbol{n}\right), \boldsymbol{\phi} \right\rangle da. \label{CurlB}
\end{equation}

\noindent (b) If  $\boldsymbol{C} \in  \mathcal{C}(\Omega,\mathbb{R}^3)$ then
\begin{equation}
\label{CurlS}
\Curl \boldsymbol{C} \left(\boldsymbol{\phi}\right)=\int_{\partial S-\partial \Omega} \langle \boldsymbol{c}\times \boldsymbol{\nu},\boldsymbol{\phi} \rangle dl +\int_S  \langle \left(-\kappa\boldsymbol{c}\times \boldsymbol{n}+ \curl_S \boldsymbol{c}\right),\boldsymbol{\phi} \rangle da + \int_S \left\langle \boldsymbol{c} \times \boldsymbol{n} ,\frac{\partial \boldsymbol{\phi}}{\partial n}   \right\rangle da.
\end{equation}

\noindent (c) If $\boldsymbol{F} \in \mathcal{F}(\Omega,\mathbb{R}^3)$ then
\begin{equation}
\label{CurlF}
\begin{split}
\Curl \boldsymbol{F} (\boldsymbol{\phi}) = \int_{\partial S - \partial \Omega} \left\langle\left( \boldsymbol{f}\times \boldsymbol{\nu}\right), \frac{\partial \boldsymbol{\phi}}{\partial n} \right\rangle dl  +\int_{\partial S - \partial \Omega} \left\langle \left((\nabla_S \boldsymbol{n})\times \boldsymbol{f}\right)^T\boldsymbol{\nu},\boldsymbol{\phi} \right\rangle dl  - \int_S \langle \div_S \left((\nabla_S \boldsymbol{n})\times \boldsymbol{f}\right)^T,\boldsymbol{\phi} \rangle da\\+ \int_S   \left\langle \left(-\kappa \left(\boldsymbol{f}\times \boldsymbol{n}\right) +\curl_S \boldsymbol{f}\right),\frac{\partial \boldsymbol{\phi}}{\partial n} \right\rangle  da  + \int_S \langle \boldsymbol{f}\times\boldsymbol{n}, \left(\nabla(\nabla \boldsymbol{\phi})\boldsymbol{n}\otimes\boldsymbol{n}\right) \rangle da.
\end{split}
\end{equation}

\noindent (d) If  $\boldsymbol{H} \in \mathcal{H}(\Omega,\mathbb{R}^3)$ then
\begin{equation}
\Curl \boldsymbol{H} (\boldsymbol{\phi})= \int_L \langle \boldsymbol{h}, \curl_{t} \boldsymbol{\phi}  \rangle dl  - \int_L \left\langle \frac{\partial }{\partial t}(\boldsymbol{h} \times \boldsymbol{t}) ,  \boldsymbol{\phi}  \right\rangle dl + \langle \boldsymbol{h} \times \boldsymbol{t} ,  \boldsymbol{\phi}  \rangle |_{\partial L-\partial \Omega}.
\end{equation}
\end{ids}

The above identities will be used, in particular, to deduce the consequences of vanishing of the left hand sides in terms of derivatives of smooth functions. For instance, arbitrariness of $\boldsymbol{\phi}$ can be exploited in Equation \eqref{gida} to show the equivalence of  $\nabla B = 0$ with $\nabla b = \boldsymbol{0}$ in $\Omega-S$ and $\llbracket b \rrbracket = 0$ on $S$. Similarly, Equation \eqref{DivB} implies the equivalence of $\Div \boldsymbol{B} = 0$ with $\div \boldsymbol{b} = {0}$ in $\Omega-S$ and $ \langle \llbracket \boldsymbol{b} \rrbracket , \boldsymbol{n} \rangle= 0$ on $S$, and \eqref{CurlB} implies the equivalence of $\Curl \boldsymbol{B} = \boldsymbol{0}$ with $\curl \boldsymbol{b} = \boldsymbol{0}$ in $\Omega-S$ and $\llbracket \boldsymbol{b} \rrbracket \times \boldsymbol{n}= \boldsymbol{0}$ on $S$.\footnote{Given a distribution $T(\phi)=\int_\Omega b\phi dv + \int_S c \phi da$ such that $b$ is piecewise smooth (smooth in $\Omega-S$) and $c$ is a smooth function on $S$. Also, $T(\phi)=0$ for any $\phi\in \mathcal{D}(\Omega)$. At $x_0 \in \Omega-S$, if $b(x_0)=b_0 > 0$, there exists a connected set $A\subset \Omega-S$ with non zero volume such that $b \neq 0$ in $A$. There also exists a connected set $A_1\subset A$ such that $A_1$ has a finite volume $V_1$ with $x_0\in A_1$ and $b(x) > {b_0}/{2}$ for all $x\in A_1$. We choose $\phi\in \mathcal{D}(\Omega)$ such that $\phi(x)=1$ for all $x\in A_1$, $\phi(x)\geq 0$ for all $x\in A$, and $\phi(x)=0$ for $x \notin A$. Then $T(\phi)\geq {b_0 V_1}/{2}$ ($b$ and $\phi$ do not change signs) which gives us a contradiction. So $b=0$ for all $x\in \Omega-S$. The assumed sign of $b_0$ is clearly of no consequence. A similar argument can be constructed to argue that $c=0$.}  To establish similar results from other identities we need the following two results. 
First, if $K\in\mathcal{D}'(\Omega)$ is such that, for any $\phi \in \mathcal{D}(\Omega)$, 
\begin{equation}
K(\phi)=\int_S a \phi da + \int_S b \frac{\partial \phi}{\partial n} da +\int_S c \langle \nabla(\nabla \phi), \boldsymbol{n}\otimes \boldsymbol{n} \rangle da,
\end{equation}
where $a$, $b$, $c$ are smooth functions on the oriented regular surface $S\subset \Omega$ with normal $\boldsymbol{n}$, then $K=0$ is equivalent to $a=0$, $b=0$, and $c=0$. 
Indeed, let $(x_1,x_2,x_3)$ be a local orthogonal coordinate system with $(\boldsymbol{e}_1,\boldsymbol{e}_2, \boldsymbol{e}_3)$ as basis vectors such that $x_3=0$ defines $S$ (locally) with $\boldsymbol{n} = \boldsymbol{e}_3$. Let $\overline{\boldsymbol{n}}$ be a smooth extension of $\boldsymbol{n}$ to $\Omega$ such that $\langle \overline{\boldsymbol{n}},\overline{\boldsymbol{n}}\rangle=1$. Then $\langle \nabla(\nabla \phi),\boldsymbol{n}\otimes \boldsymbol{n} \rangle= ({\partial^2 \phi}/{\partial x_3^2}) - \langle \nabla_S \phi, ({\partial \overline{\boldsymbol{n}}}/{\partial x_3}) \rangle$. Let $f$ be an arbitrary smooth function on $S$ with a compact support $A\subset S$. Let $l$ be the minimum distance of $A$ from $\partial \Omega$. Let $B \subset \Omega$ such that $x\in B$ if and only if $\text{dist}(x,S)< l_1$, where $l_1<l$. There always exist a $g \in \mathcal{D}(\Omega)$ such that $g(x)=1$ for $x\in B$. Then for $\phi={fgx_3^2}$, $\phi=0$ and $({\partial \phi}/{\partial x_3})=0$ on $S$, and hence $\int_S c f da=0$ for an arbitrary local smooth function $f$. This implies $c=0$. Similarly, use $\phi=fgx_3$ to conclude that $b=0$ and consequently $a=0$.  
Second, if $K \in \mathcal{D}'(\Omega)$ is such that, for any $\phi \in \mathcal{D}(\Omega)$,
\begin{equation}
 K(\phi)=\int_L a \phi dl + \int_L \langle \boldsymbol{b},(\boldsymbol{I}-\boldsymbol{t}\otimes \boldsymbol{t})\nabla \phi \rangle dl,
\end{equation}
where $a$ and $\boldsymbol{b}$ are smooth functions on a smooth oriented curve $L\subset \Omega$ with tangent $\boldsymbol{t}$. Then $K=0$ is equivalent to $a=0$ and $(\boldsymbol{I}-\boldsymbol{t}\otimes \boldsymbol{t})\boldsymbol{b}=\boldsymbol{0}$. Indeed,
let $(x_1,x_2,x_3)$ be a local orthogonal coordinate system with $(\boldsymbol{e}_1,\boldsymbol{e}_2,\boldsymbol{e}_3)$ as basis vectors such that $L$ is locally parameterized by $x_3$, i.e. $\boldsymbol{t} = \boldsymbol{e}_3$, $x_1=0$, and $x_2=0$ on $L$. By considering $\phi$ in terms of an arbitrary smooth function, with local compact support on $L$, in addition to being linear in $x_1$ and $x_2$, we can use arguments analogous to the previous paragraph to derive the required results. 

A direct application of the above results, in conjunction with Equation \eqref{DivC} is the equivalence of $\Div \boldsymbol{C} = 0$ with $\div_S \boldsymbol{c} = {0}$ and $\langle \boldsymbol{c},\boldsymbol{n}\rangle = {0}$ in $S$ and $\langle \boldsymbol{c},\boldsymbol{\nu}\rangle = {0}$ on $\partial S - \partial \Omega$. Similarly, Equation \eqref{CurlS} implies the equivalence of $\Curl \boldsymbol{C} = \boldsymbol{0}$ with $\curl_S \boldsymbol{c} = \boldsymbol{0}$ and $\boldsymbol{c} \times \boldsymbol{n} = \boldsymbol{0}$ in $S$ and $\boldsymbol{c} \times \boldsymbol{\nu} = \boldsymbol{0}$ on $\partial S - \partial \Omega$. Furthermore, Equation \eqref{DivF} would imply the equivalence of $\Div \boldsymbol{F} = 0$ with $\div_S \boldsymbol{f} = {0}$, $\div_S ((\nabla_S \boldsymbol{n})\boldsymbol{f}) = {0}$, and $\langle \boldsymbol{f},\boldsymbol{n}\rangle = {0}$ in $S$, and $\langle \boldsymbol{f},\boldsymbol{\nu}\rangle = {0}$, $\langle(\nabla_S \boldsymbol{n})\boldsymbol{f},\boldsymbol{\nu}\rangle = 0$ on $\partial S - \partial \Omega$. Analogous consequences can be deduced from other identities.

\subsection{Poincar\'e's lemma}

Given any $U \in \mathcal{D}'(\Omega)$ and $\boldsymbol{V} \in \mathcal{D}'(\Omega,\mathbb{R}^3)$, 
\begin{equation}
\Curl \left( \nabla U \right) = \boldsymbol{0}~\text{and}~
\Div \left( \Curl \boldsymbol{V}\right) = 0.
\label{poin1}
\end{equation}
These follow immediately by writing $\left(\Curl \left( \nabla U\right)\right) _i = \epsilon_{ijk} \partial^2_{jk} U$ and 
$\Div \left( \Curl \boldsymbol{V}\right) =\epsilon_{ijk}\partial^2_{ik} V_j$ and recalling 
Equation \eqref{diff_sym}. The converse of these results is less straightforward. 
The following theorem, stated by Mardare \cite{mardare2008poincare} in this form, establishes that the converse of \eqref{poin1}$_1$ holds true for a simply connected domain in the case of curl free vector valued distributions. For a proof, we refer the reader to the original paper. 
\begin{theorem} (Mardare, 2008 \cite{mardare2008poincare})
\label{PoincareOneForm}
If $\Omega$ is a simply connected open subset of $\mathbb{R}^3$ and $\boldsymbol{V} \in \mathcal{D}' (\Omega,\mathbb{R}^3)$, such that $\Curl \boldsymbol{V} =\boldsymbol{0}$, 
then there exist a $U \in \mathcal{D}'(\Omega)$ such that 
$\boldsymbol{V} = \nabla U$.
\end{theorem}
An immediate corollary of Theorem \ref{PoincareOneForm} is to establish an analogous result for symmetric tensor valued distributions.
\begin{cor}
\label{StrainCompatibility}
If $\Omega$ is a simply connected open subset of $\mathbb{R}^3$ and $\boldsymbol{A} \in \mathcal{D}' (\Omega,\Sym)$, then
$\Curl \Curl \boldsymbol{A}=\boldsymbol{0}$
is equivalent to existence of a $\boldsymbol{U} \in \mathcal{D}' (\Omega, \mathbb{R}^3)$ such that 
$\boldsymbol{A}= (1/2)({\nabla \boldsymbol{U} +(\nabla \boldsymbol{U})^T})$.
\end{cor}
\begin{proof}
Let $H_{ijk}\in \mathcal{D}' (\Omega)$ be such that
$H_{ijk}={\partial_j A_{ik}} - {\partial_i A_{jk}}$.
Then,  
${\partial_l H_{ijk}} - {\partial_k H_{ijl}}=0$ which, according to
Theorem \ref{PoincareOneForm}, implies the existence of $P_{ij} \in \mathcal{D}'(\Omega)$ such that
$H_{ijk} = {\partial_k P_{ij}}$.
Since $H_{ijk}=-H_{jik}$, or equivalently $\partial_k (P_{ij} + P_{ji})=0$, we can always construct a $P_{ij}$ such that $P_{ij} + P_{ji}=0$ and ${\partial_k P_{ij}}=H_{ijk}$.
Let $Q_{ij}=A_{ij}+P_{ij}$. Then
${\partial_k Q_{ij}} - {\partial_j Q_{ik}}=0$ and, as a consequence of Theorem \ref{PoincareOneForm}, there exist a $\boldsymbol{U} \in \mathcal{D}'(\Omega,\mathbb{R}^3)$, such that $Q_{ij} = {\partial_j U_{i}}$. The converse can be established using Equation \eqref{diff_sym}. 
\end{proof}

It should be noted that both Theorem \ref{PoincareOneForm} and Corollary \ref{StrainCompatibility} do not establish any regularity on distributions $U$ and $\boldsymbol{U}$, respectively, if we were to start with assuming certain regularity on distributions $\boldsymbol{V}$ and $\boldsymbol{A}$. For instance, if we start with an $\boldsymbol{A}$ in $\mathcal{B} (\Omega,\Sym)$ then what distribution space should $\boldsymbol{U}$ belong to? We will answer several such questions in Section \ref{rr}.

The next theorem proves the converse of \eqref{poin1}$_2$ for divergence free vector valued distributions on a contractible domain. Our proof, whose major part appears in Appendix \ref{poinapp}, is adapted from a more general proof given by Demailly \cite[p. 20]{demailly1997complex} within the framework of currents. Currents on open sets in $\mathbb{R}^3$ correspond to vector valued distributions, in a manner similar to the correspondence of smooth forms to smooth vector fields \cite{de2012differentiable} .
\begin{theorem}
\label{Poincare}
If $\Omega$ be a contractible open set of $\mathbb{R}^3$ and $\boldsymbol{T}\in \mathcal{D}'(\Omega,\mathbb{R}^3)$, such that $\Div \boldsymbol{T} =0$, then there exist a $\boldsymbol{S}\in \mathcal{D}'(\Omega,\mathbb{R}^3)$ such that $\boldsymbol{T} = \Curl \boldsymbol{S}$.
\begin{proof}
According to Lemma (\ref{CohomologoustoSmoothForm}) we have $\boldsymbol{u} \in C^{\infty}(\Omega , \mathbb{R}^3)$ and $\boldsymbol{S_1}\in \mathcal{D}'(\Omega,\mathbb{R}^3)$ such that 
$\boldsymbol{T_u} - \boldsymbol{T} = \Curl \boldsymbol{S_1}$.
We use $\Div(\Curl \boldsymbol{S_1}) = 0$ and $\Div \boldsymbol{T} = 0$ to obtain $\Div \boldsymbol{T_u} =0$ which implies  $\div \boldsymbol{u} = 0$. According to Poincare's lemma for smooth vector fields \cite{do2012differential}, there then exists  $\boldsymbol{\omega} \in C^\infty (\Omega,\mathbb{R}^3)$ such that $\curl \boldsymbol{\omega} = \boldsymbol{u}$.
Consequently,
$\boldsymbol{T}= \boldsymbol{T}_{\curl \boldsymbol{\omega}} - \Curl \boldsymbol{S_1}= \Curl \boldsymbol{T}_{\boldsymbol{\omega}} - \Curl \boldsymbol{S_1}= \Curl (\boldsymbol{T}_{\boldsymbol{\omega}}-\boldsymbol{S_1})$,
thereby proving our assertion. 
\end{proof}
\end{theorem}

\begin{rem}
The above results are well known in the context of smooth fields. In particular, in the language of differential forms \cite{do2012differential}, for any smooth form $\omega$, $d(d \omega)=0$, where $d$ denotes the exterior derivative. For differential forms of degree 0, 1 and 2, the exterior derivative corresponds to gradient, curl, and divergence operator, respectively.
Moreover, for any smooth p-form $\omega$ on a contractible domain such that $d \omega =0$, there exist a (p-1)-form $\omega_1$ such that $\omega = d\omega_1$. For a 1-form, this result holds even for simply connected domains. Our assertions extend these results to a more general situation where the components of the vector fields are distributions instead of smooth functions.
\end{rem}

\subsection{Regularity Results}
\label{rr}
In this section, we collect several results of the kind mentioned in Theorem \ref{PoincareOneForm} and Corollary \ref{StrainCompatibility}, but restrict ourselves to specific subsets of distributions. In Lemma \ref{RegularityLemma} below, we start with curl free vector valued distributions, defined in terms of elements from $\mathcal{B}(\Omega,\mathbb{R}^3)$, $\mathcal{C}(\Omega,\mathbb{R}^3)$, and $\mathcal{F}(\Omega,\mathbb{R}^3)$, and determine the precise form of distributions whose gradients are equal to the vector valued distributions. 

The spaces $\mathcal{B}(\Omega)$, $\mathcal{C}(\Omega)$, $\mathcal{B}(\Omega,\mathbb{R}^3)$, $\mathcal{C}(\Omega,\mathbb{R}^3)$ and $\mathcal{F}(\Omega,\mathbb{R}^3)$, used in the following, are as defined in Equations \eqref{distB}, \eqref{distC}, and \eqref{vectdist}.

\begin{lemma}
\label{RegularityLemma}
Let $\Omega \subset \mathbb{R}^3$ be a simply connected region  and $S \subset \Omega$ be a regular oriented surface such that $\partial S-\partial \Omega=\emptyset$. Then, for $\psi \in \mathcal{D}(\Omega)$ and $\boldsymbol{\phi}\in \mathcal{D}(\Omega,\mathbb{R}^3)$,

\noindent (a)  The condition $\Curl \boldsymbol{C}=\boldsymbol{0}$, with $\boldsymbol{C} \in \mathcal{C}(\Omega,\mathbb{R}^3)$, is equivalent to existence of a $U \in \mathcal{B}(\Omega)$ such that $\boldsymbol{C}=\nabla U$.

\noindent (b) The condition $\Curl \boldsymbol{T}=\boldsymbol{0}$, with $\boldsymbol{T} \in \mathcal{D}'(\Omega, \mathbb{R}^3)$ and $\boldsymbol{T}(\boldsymbol{\phi})=\boldsymbol{B}(\boldsymbol{\phi})+\boldsymbol{C}(\boldsymbol{\phi})$, where $\boldsymbol{B}\in \mathcal{B}(\Omega,\mathbb{R}^3)$ and $\boldsymbol{C}\in \mathcal{C}(\Omega,\mathbb{R}^3)$, is equivalent to existence of a $U \in \mathcal{B}(\Omega)$ such that $\boldsymbol{T}=\nabla U$.

\noindent (c) The condition $\Curl \boldsymbol{T}=\boldsymbol{0}$, with $\boldsymbol{T} \in \mathcal{D}'(\Omega, \mathbb{R}^3)$ and $\boldsymbol{T}(\boldsymbol{\phi})=\boldsymbol{B}(\boldsymbol{\phi})+\boldsymbol{C}(\boldsymbol{\phi})+\boldsymbol{F}(\boldsymbol{\phi})$, where $\boldsymbol{B}\in \mathcal{B}(\Omega,\mathbb{R}^3)$, $\boldsymbol{C}\in \mathcal{C}(\Omega,\mathbb{R}^3)$, and $\boldsymbol{F}\in \mathcal{F}(\Omega,\mathbb{R}^3)$, is equivalent to existence of a $U \in \mathcal{D}'(\Omega)$ such that $U(\psi)=B(\psi)+C(\psi)$, where $B\in \mathcal{B}(\Omega)$ and $C\in \mathcal{C}(\Omega)$, with  $\boldsymbol{T}=\nabla U$.
\end{lemma}

\begin{proof} The existence of a $U \in \mathcal{D}'(\Omega)$ is guaranteed in all the above cases by Theorem \ref{PoincareOneForm}. Our goal is to however establish a stricter regularity on $U$ for the given conditions. That $\partial S-\partial \Omega=\emptyset$ implies that $S$ divides $\Omega$ into mutually exclusive open sets $\Omega^+$ and $\Omega^-$ such that $\partial \Omega^+ \cap \partial \Omega^- = S$ and $\Omega^+ \cup S \cup \Omega^-=\Omega$. 

\noindent (a) According to Identity \eqref{CurlS}, $\Curl \boldsymbol{C}=\boldsymbol{0}$ is equivalent to $\boldsymbol{c}\times \boldsymbol{n}=\boldsymbol{0}$ and $\curl_S \boldsymbol{c}=\boldsymbol{0}$. Hence $\boldsymbol{c}=c_0\boldsymbol{n}$, for a fixed $c_0 \in \mathbb{R}$. Then $U \in \mathcal{B}(\Omega)$ such that $U(\psi)=\int_{\Omega} b_0 \psi dv$, where $b_0 = c_0$ in $\Omega^-$ and $0$ in $\Omega^+$, satisfies $\boldsymbol{C} = \nabla U$.

\noindent(b) According to Identities \eqref{CurlB} and \eqref{CurlS}, $\Curl \boldsymbol{T}=\boldsymbol{0}$ implies $\boldsymbol{c}\times \boldsymbol{n}=\boldsymbol{0}$, which is equivalent to $\boldsymbol{c}=c\boldsymbol{n}$, $\curl \boldsymbol{b}=\boldsymbol{0}$ in $\Omega-S$, and $(\llbracket \boldsymbol{b} \rrbracket-\nabla_S c)\times \boldsymbol{n}=\boldsymbol{0}$ on $S$. The second equation is equivalent to existence of a ${u}:\Omega \to \mathbb{R}$ such that ${u}|_{\Omega^+} \in C^{\infty}(\Omega^+)$, ${u}|_{\Omega^-} \in C^{\infty}(\Omega^-)$, and $\nabla u=\boldsymbol{b}$ in $\Omega-S$, cf. \cite{Kinematics}. We introduce $U_1 \in \mathcal{B}(\Omega)$ such that $U_1(\phi)=\int_\Omega u \phi dv$. Then, using Equation \eqref{gida}, we get $\nabla U_1(\boldsymbol{\phi}) =\int_\Omega \langle \boldsymbol{b}, \boldsymbol{\phi }\rangle dv- \int_S \langle \llbracket u\rrbracket \boldsymbol{n},\boldsymbol{\phi} \rangle da$. Consequently, $(\boldsymbol{T}-\nabla U_1)=\int_S \langle (\llbracket u\rrbracket \boldsymbol{n} + \boldsymbol{c}),\boldsymbol{\phi} \rangle da$. Noting that $\Curl(\boldsymbol{T}-\nabla U_1)=\boldsymbol{0}$, in conjunction with part (a) of the lemma, we have a $U_2 \in \mathcal{B}(\Omega)$ such that $ \boldsymbol{T}-\nabla U_1 = \nabla U_2$.  The required $U\in \mathcal{B}(\Omega)$ is given by $U=U_1+U_2$.

\noindent (c) According to Identity \eqref{CurlF}, $\Curl \boldsymbol{T}=\boldsymbol{0}$ implies $\boldsymbol{f}\times \boldsymbol{n}=\boldsymbol{0}$ or, equivalently,  that $\boldsymbol{f}=f\boldsymbol{n}$, where $f \in C^{\infty}(S)$. We introduce $U_1 \in \mathcal{C}(\Omega)$ such that $U_1(\psi)=-\int_{S} f \psi da$. Then, using Equation \eqref{gidb}, we get $\nabla U_1(\boldsymbol{\phi})=-\int_S \langle (\nabla_S f + \kappa f \boldsymbol{n}),\boldsymbol{\phi}\rangle da+\int_S \langle f \boldsymbol{n},(\partial\boldsymbol{\phi}/{\partial n})\rangle da$. Consequently, $(\boldsymbol{T}-\nabla U_1)(\boldsymbol{\phi})=\boldsymbol{B}(\boldsymbol{\phi})+\boldsymbol{C}(\boldsymbol{\phi})+\int_S \langle (\kappa f \boldsymbol{n}+\nabla_S f),\boldsymbol{\phi}\rangle da$. Noting that $\Curl (\boldsymbol{T}-\nabla U_1)=\boldsymbol{0}$, in conjunction with part (a) of the lemma, we have a $U_2 \in \mathcal{B}(\Omega)$ such that $\nabla U_2=\boldsymbol{T}-\nabla U_1$. The required distribution is given by $U=U_1+U_2$.

\noindent The converse in all the above results follows from Equation \eqref{diff_sym} in a straightforward manner.
\end{proof}

In Corollaries \ref{StrainCompatibilityB} and \ref{StrainCompatibilityBC}, we revisit Corollary \ref{StrainCompatibility} in the light of the above lemma but assume $\boldsymbol{A}$ to be in terms of elements from $\mathcal{B} (\Omega,\Sym)$ and $\mathcal{C} (\Omega,\Sym)$ and determine the precise form of $\boldsymbol{U}$. These regularity results are motivated from their applicability in deriving strain compatibility relations in Section \ref{cdsf}.

\begin{cor}
\label{StrainCompatibilityB}
If $\Omega$ is a simply connected open subset of $\mathbb{R}^3$ and $\boldsymbol{A} \in \mathcal{B} (\Omega,\Sym)$, then
$\Curl \Curl \boldsymbol{A}=\boldsymbol{0}$
is equivalent to existence of a $\boldsymbol{U} \in \mathcal{B} (\Omega, \mathbb{R}^3)$, with $\boldsymbol{U}(\boldsymbol{\phi})=\int_\Omega \langle \boldsymbol{u},\boldsymbol{\phi}\rangle dv$, where $\boldsymbol{u}$ is a piecewise smooth vector field continuous across $S$, such that 
$\boldsymbol{A}= (1/2)({\nabla \boldsymbol{U} +(\nabla \boldsymbol{U})^T})$.
\end{cor}
\begin{proof}
Let $H_{ijk}\in \mathcal{D}' (\Omega)$ be given as
$H_{ijk}={\partial_j A_{ik}} - {\partial_i A_{jk}}$. Then, on one hand,  Identity \eqref{gida} implies $H_{ijk}(\psi)=B(\psi)+C(\psi)$, for $\psi \in \mathcal{D}(\Omega)$,  where $B\in \mathcal{B}(\Omega)$ and $C\in \mathcal{C}(\Omega)$. On the other hand, we have ${\partial_l H_{ijk}} - {\partial_k H_{ijl}}=0$ which, according to
Lemma \ref{RegularityLemma}(b), posits the existence of $P_{ij} \in \mathcal{B}(\Omega)$ such that
$H_{ijk} = {\partial_k P_{ij}}$.
Moreover, since $H_{ijk}=-H_{jik}$, or equivalently $\partial_k (P_{ij} + P_{ji})=0$, we can always construct a $P_{ij}$ such that $P_{ij} + P_{ji}=0$ and ${\partial_k P_{ij}}=H_{ijk}$.
Let $Q_{ij}=A_{ij}+P_{ij}$. Then
${\partial_k Q_{ij}} - {\partial_j Q_{ik}}=0$ and, as a consequence of Lemma \ref{RegularityLemma}(a), there exist a $\boldsymbol{U} \in \mathcal{B}(\Omega,\mathbb{R}^3)$, such that $Q_{ij} = {\partial_j U_{i}}$. We can write $\boldsymbol{U}(\boldsymbol{\phi})=\int_\Omega \langle \boldsymbol{u},\boldsymbol{\phi}\rangle dv$, where $\boldsymbol{u}$ is a piecewise smooth vector field on $\Omega$. Using identity \eqref{gida} we have $((1/2)({\nabla \boldsymbol{U} +(\nabla \boldsymbol{U})^T})) (\boldsymbol{\psi})=\boldsymbol{B}_1 (\boldsymbol{\psi})+\int_S \langle ((1/2)(\llbracket\boldsymbol{u}\rrbracket\otimes \boldsymbol{n}+\boldsymbol{n}\otimes \llbracket\boldsymbol{u}\rrbracket)), \boldsymbol{\psi} \rangle da$, for all $\boldsymbol{\psi}\in \mathcal{D}(\Omega,\Lin)$, where $\boldsymbol{B}_1\in \mathcal{B}(\Omega,\Sym)$. Since $\boldsymbol{A}$ has no surface concentration, we require $\llbracket\boldsymbol{u}\rrbracket=\boldsymbol{0}$. The converse follows from Equation \eqref{diff_sym}.
\end{proof}

\begin{cor}
\label{StrainCompatibilityBC}
If $\Omega$ is a simply connected open subset of $\mathbb{R}^3$ and $\boldsymbol{A} \in \mathcal{D}' (\Omega,\Sym)$, which, for $\boldsymbol{\phi} \in \mathcal{D} (\Omega,\Lin)$, is given as $\boldsymbol{A}(\boldsymbol{\phi})=\boldsymbol{B}(\boldsymbol{\phi})+\boldsymbol{C}(\boldsymbol{\phi})$, where $\boldsymbol{B} \in \mathcal{B} (\Omega,\Sym)$ and $\boldsymbol{C} \in \mathcal{C} (\Omega,\Sym)$, then
$\Curl \Curl \boldsymbol{A}=\boldsymbol{0}$ 
is equivalent to existence of a $\boldsymbol{U} \in \mathcal{B} (\Omega, \mathbb{R}^3)$ such that 
$\boldsymbol{A}= (1/2)({\nabla \boldsymbol{U} +(\nabla \boldsymbol{U})^T})$.
\end{cor}
\begin{proof}
Let $H_{ijk}\in \mathcal{D}' (\Omega)$ be given as
$H_{ijk}={\partial_j A_{ik}} - {\partial_i A_{jk}}$. Then, on one hand,  Identities \eqref{gida} and \eqref{gidb} imply that $H_{ijk}(\psi)=B(\psi)+C(\psi) + F(\psi)$, for $\psi \in \mathcal{D}(\Omega)$,  where $B\in \mathcal{B}(\Omega)$, $C\in \mathcal{C}(\Omega)$, and $F\in \mathcal{F}(\Omega)$. On the other hand, we have ${\partial_l H_{ijk}} - {\partial_k H_{ijl}}=0$ which, according to
Lemma \ref{RegularityLemma}(c), posits the existence of $P_{ij} \in \mathcal{D}'(\Omega)$ with $P_{ij}(\psi)=B(\psi)+C(\psi)$, for $\psi \in \mathcal{D}(\Omega)$, such that
$H_{ijk} = {\partial_k P_{ij}}$.
Moreover, since $H_{ijk}=-H_{jik}$, or equivalently $\partial_k (P_{ij} + P_{ji})=0$, we can always construct a $P_{ij}$ such that $P_{ij} + P_{ji}=0$ and ${\partial_k P_{ij}}=H_{ijk}$.
Let $Q_{ij}=A_{ij}+P_{ij}$. Then
${\partial_k Q_{ij}} - {\partial_j Q_{ik}}=0$ and, as a consequence of Lemma \ref{RegularityLemma}(b), there exist a $\boldsymbol{U} \in \mathcal{B}(\Omega,\mathbb{R}^3)$, such that $Q_{ij} = {\partial_j U_{i}}$. The converse follows from Equation \eqref{diff_sym}.
\end{proof}

\begin{rem}
It is pertinent here to note some existing literature on such regularity results. Amrouche and Girault \cite{amrouche1994decomposition} have shown that, given a distribution $U\in \mathcal{D}'(\Omega)$, $\nabla U \in H^{-m}(\Omega,\mathbb{R}^3)$ implies that $U\in H^{-m+1}(\Omega)$,  where $H^{-m}(\Omega)$, for non-negative integer $m$, is the dual of $H^{m}_0(\Omega)$, the latter being the usual Sobolev space. Amrouche et. al. \cite{amrouche2006characterizations} have generalised this result to show that, for a vector valued distribution $\boldsymbol{U}\in \mathcal{D}'(\Omega,\mathbb{R}^3)$, $(1/2(\nabla\boldsymbol{U}+(\nabla\boldsymbol{U})^T))\in H^{-m}(\Omega,\Sym)$ implies that $\boldsymbol{U}\in H^{-m+1}(\Omega,\mathbb{R}^3)$.
\end{rem}


\section{Compatibility of discontinuous strain fields} 
\label{cdsf}

This section is divided into two parts. In the first, we consider a piecewise smooth symmetric tensor field over a simply connected $\Omega$ and obtain the necessary and sufficient conditions for there to exist a piecewise smooth, but continuous, vector field over $\Omega$, the symmetric part of whose gradient is equal to the tensor field away from the surface of discontinuity. This is tantamount to seeking conditions on the  piecewise smooth strain tensor field, possibly discontinuous over a surface $S\subset\Omega$, such that it is obtainable from a piecewise smooth, but continuous, displacement vector field as the symmetric part of its gradient (away from $S$). This is the well known problem of strain compatibility. Whereas the conditions on a smooth strain field are routinely derived in books on elasticity, the jump conditions, necessary to enforce compatibility of strain across the surface of discontinuity, have been discussed rarely and only in specific forms \cite{markenscoff1996note, wheeler1999conditions}. These conditions, in their most general form, are obtained in Section \ref{pcs} below using the preceding mathematical infrastructure. We also reduce our general conditions to those available in literature. In the second part, in Section \ref{cs}, we revisit the problem of strain compatibility after relaxing the requirement for continuity of displacement field across $S$, thereby allowing the interface to be imperfectly bonded. As we shall see below, such a framework necessarily requires us to consider a strain field, concentrated over $S$, in addition to a piecewise smooth strain field in the bulk. 

\subsection{Perfectly Bonded Surface of Discontinuity}
\label{pcs}

Let $\boldsymbol{e}$ be a piecewise smooth symmetric tensor field on a simply connected domain $\Omega$, possibly discontinuous across a regular oriented surface $S \in \Omega$ with $\partial S - \partial \Omega = \emptyset$. Then, for a compactly supported smooth tensor valued field $\boldsymbol{\phi} \in \mathcal{D}({\Omega, \Lin})$, we can define a distribution $\boldsymbol{E} \in \mathcal{B}(\Omega,\Sym)$ such that
\begin{equation}
\boldsymbol{E}(\boldsymbol{\phi})= \int_{\Omega} \langle \boldsymbol{e},\boldsymbol{\phi} \rangle dv. \label{diststrain}
\end{equation}
Using Identity \eqref{CurlB}, we can write
\begin{equation}
\Curl \boldsymbol{E} (\boldsymbol{\phi})=\int_{\Omega} \langle \curl \boldsymbol{e}, \boldsymbol{\phi} \rangle dv +\int_S \langle (\llbracket \boldsymbol{e} \rrbracket \times \boldsymbol{n})^T, \boldsymbol{\phi}\rangle da.
\end{equation} 
Clearly, $\Curl \boldsymbol{E}$ is composed of distributions $\boldsymbol{B} \in \mathcal{B}(\Omega,\Lin)$ and $\boldsymbol{C}  \in \mathcal{C}(\Omega,\Lin)$ such that $\boldsymbol{B}(\boldsymbol{\phi})=\int_{\Omega} \langle \curl \boldsymbol{e},\boldsymbol{\phi} \rangle dv$ and $\boldsymbol{C}(\boldsymbol{\phi})=\int_S \langle (\llbracket \boldsymbol{e} \rrbracket\times \boldsymbol{n})^T , \boldsymbol{\phi} \rangle  da$.
According to Identities \eqref{CurlB} and \eqref{CurlS}, we have
\begin{equation}
\label{A}
\Curl \boldsymbol{B} (\boldsymbol{\phi}) = \int_{\Omega} \langle\curl \curl \boldsymbol{e}, \boldsymbol{\phi} \rangle dv + \int_S \langle (\llbracket \curl \boldsymbol{e}\rrbracket \times \boldsymbol{n} )^T, \boldsymbol{\phi} \rangle da ~\text{and} \nonumber
\end{equation}
\begin{equation}
\Curl \boldsymbol{C}(\boldsymbol{\phi})=\int_S \left( \left\langle -\kappa  \left((\llbracket \boldsymbol{e} \rrbracket\times \boldsymbol{n})^T\times \boldsymbol{n} \right)^T +\curl_S(\llbracket \boldsymbol{e} \rrbracket\times \boldsymbol{n})^T,\boldsymbol{\phi} \right\rangle + \left\langle \left( (\llbracket \boldsymbol{e} \rrbracket\times \boldsymbol{n})^T\times \boldsymbol{n}\right)^T,\frac{\partial \boldsymbol{\phi}}{\partial n} \right\rangle \right) da, \nonumber
\end{equation}
respectively, allowing us to obtain $\Curl \Curl \boldsymbol{E}= \Curl \boldsymbol{B}  + \Curl \boldsymbol{C}$.
The condition $\Curl \Curl \boldsymbol{E} (\boldsymbol{\phi})=\boldsymbol{0}$, for arbitrary $\boldsymbol{\phi}$, is therefore equivalent to requiring
 \begin{eqnarray}
 \label{BulkCompatibility}
 \curl \curl \boldsymbol{e}=\boldsymbol{0}~\text{in}~\Omega - S, \\
\label{I1}
\left( (\llbracket \boldsymbol{e} \rrbracket\times \boldsymbol{n})^T\times \boldsymbol{n}\right)^T=\boldsymbol{0} ~\text{on}~S,~ \text{and}
\\
\label{I2}
(\llbracket \curl \boldsymbol{e} \rrbracket \times \boldsymbol{n} )^T+\curl_S (\llbracket \boldsymbol{e} \rrbracket\times \boldsymbol{n})^T=\boldsymbol{0}~\text{on}~S.
\end{eqnarray}
On the other hand, according to Corollary \ref{StrainCompatibilityB}, $\Curl \Curl \boldsymbol{E}=  \boldsymbol{0}$, with $\boldsymbol{E}$ given by \eqref{diststrain}, is equivalent to existence of a $\boldsymbol{U} \in \mathcal{B}(\Omega, \mathbb{R}^3)$ such that $\boldsymbol{E}= (1/2)({\nabla \boldsymbol{U} +(\nabla \boldsymbol{U})^T})$, with $\boldsymbol{U}(\boldsymbol{\psi})=\int_\Omega \langle \boldsymbol{u},\boldsymbol{\psi}\rangle dv$, for $\boldsymbol{\psi}\in \mathcal{D}(\Omega,\mathbb{R}^3)$, where $\boldsymbol{u}$ is a piecewise smooth vector field continuous across $S$. 
Summarizing the above, we have
\begin{proposition}
For a piecewise smooth tensor valued field $\boldsymbol{e}$, on a simply connected domain $\Omega \subset \mathbb{R}^3$, allowed to be discontinuous across an oriented regular surface $S \subset \Omega$ with unit normal $\boldsymbol{n}$ and  $\partial S - \partial \Omega = \emptyset$, there exists a piecewise smooth vector valued field $\boldsymbol{u}$ on $\Omega$, continuous across $S$,  such that $\boldsymbol{e}=(1/2)(\nabla \boldsymbol{u}+(\nabla \boldsymbol{u})^T)$ on $\Omega - S$ if and only if $\boldsymbol{e}$ satisfies Equations \eqref{BulkCompatibility}, \eqref{I1}, and \eqref{I2}.
\end{proposition}

In the rest of this subsection, we will use a series of remarks to discuss compatibility equations \eqref{BulkCompatibility}-\eqref{I2}. In particular, we will reduce them to forms previously derived in literature \cite{markenscoff1996note, wheeler1999conditions}. as well as connect them to certain related results by Ciarlet and Mardare \cite{ciarlet2014intrinsic} on obtaining strain compatibility relations which are equivalent to prescribing displacement boundary conditions.

\begin{rem} (Planar strain field) Let $P \in \mathbb{R}^3$ be a plane spanned by $\boldsymbol{e}_1$ and $\boldsymbol{e}_2$, with $\boldsymbol{e}_3$ as the normal to the plane, where $(\boldsymbol{e}_1, \boldsymbol{e}_2, \boldsymbol{e}_3)$ form a fixed orthonormal basis for $\mathbb{R}^3$.  The intersection of surface $S$ with plane $P$ is a planar curve $C$ with unit tangent $\boldsymbol{t}$, in plane normal $\boldsymbol{n}$, and curvature $k$. We call a distribution $\boldsymbol{E} \in \mathcal{B}(\Omega,\Sym)$ planar if $E_{ij}=0$, for $i=3$ or $j=3$, and $\partial_{3} \boldsymbol{E}=\boldsymbol{0}$. For planar $\boldsymbol{E}$, $\Curl \Curl \boldsymbol{E}$ has only one non-zero component, $\langle \Curl \Curl \boldsymbol{E}, \boldsymbol{e}_3 \otimes \boldsymbol{e}_3 \rangle$. The condition $\Curl \Curl \boldsymbol{E}=\boldsymbol{0}$ therefore reduces to one scalar equation, 
$\partial^2_{11} E_{22} +\partial^2_{22} E_{11}  - 2 \partial^2_{12} E_{12}=0$. On the other hand, the three compatibility equations \eqref{BulkCompatibility}-\eqref{I2} are reduced to
\begin{eqnarray}
\label{planarcomp1}
\frac{\partial ^2 e_{22}}{\partial x_1 ^2}+\frac{\partial ^2 e_{11}}{\partial x_2 ^2}-2\frac{\partial ^2 e_{12}}{\partial x_1 \partial x_2}=0~\text{in}~P-C, \\
\llbracket e_{ij} \rrbracket t_i t_j=0~\text{on}~C,~\text{and} \label{planarcomp2}
\\
\left\llbracket \frac{\partial e_{ij}}{\partial x_p} \right\rrbracket n_i t_j t_p + \left\llbracket \left(\frac{\partial e_{ij}}{\partial x_p} - \frac{\partial e_{pj}}{\partial x_i}\right) \right\rrbracket n_i t_j t_p + k\left\llbracket e_{ij} \right\rrbracket n_i n_j =0~\text{on}~C, \label{planarcomp3}
\end{eqnarray}
respectively.
The interfacial compatibility conditions in this form  for planar strain fields have been obtained by Markenscoff \cite{markenscoff1996note} using the continuity of displacement and its tangential derivative along the interface curve.
\label{psf}
\end{rem}

\begin{rem} (Jump conditions in an orthogonal coordinate system)
We consider an orthogonal coordinate system $(\theta_1,\theta_2,\theta_3) \in \mathbb{R}^3$, in neighborhood of $S$, and define $\boldsymbol{f}_i={\partial \boldsymbol{x}}/{\partial \theta_i}$, $f_{ii}=\langle \boldsymbol{f}_i,\boldsymbol{f}_i \rangle$ (no summation), and $\boldsymbol{\varepsilon}_i={\boldsymbol{f}_i}/{\sqrt{f_{ii}}}$ (no summation) such that $\boldsymbol{\varepsilon}_3=\boldsymbol{n}$, $\boldsymbol{\varepsilon}_1\times \boldsymbol{\varepsilon}_2=\boldsymbol{\varepsilon}_3$, and $\langle \boldsymbol{\varepsilon}_1, \boldsymbol{\varepsilon}_2\rangle=0$. We introduce $k_{\alpha}={\langle {\partial \boldsymbol{\varepsilon}_3}/{\partial \theta_\alpha}, \boldsymbol{\varepsilon}_\alpha \rangle}/{\sqrt{f_{\alpha \alpha}}}$ (no summation). The components of strain tensor $\boldsymbol{e}$ with respect to $\boldsymbol{\varepsilon}_i$-basis are $\epsilon_{ii}=\langle \boldsymbol{e}, \boldsymbol{\varepsilon}_i\otimes \boldsymbol{\varepsilon}_i \rangle$ (no summation) and $\epsilon_{ij}=2\langle \boldsymbol{e}, \boldsymbol{\varepsilon}_i\otimes \boldsymbol{\varepsilon}_j \rangle$ for $i\neq j$ (no summation). 
The jump condition \eqref{I1} is then equivalent to $\llbracket \epsilon_{\alpha \beta}\rrbracket=0$ on $S$.
On the other hand, the jump condition \eqref{I2} is equivalent to $\langle \llbracket( \curl \boldsymbol{e}  \times \boldsymbol{n} )^T+\curl_S ( \boldsymbol{e} \times \boldsymbol{n})^T\rrbracket, \boldsymbol{\varepsilon}_\beta \otimes \boldsymbol{\varepsilon}_\alpha \rangle=0$ which, using the identity
\begin{equation}
\label{curl_curvilinear}
\langle \curl \boldsymbol{e}, (\boldsymbol{w}\times \boldsymbol{v})\otimes \boldsymbol{u} \rangle= \langle \nabla \boldsymbol{e}, \boldsymbol{u}\otimes \boldsymbol{v} \otimes \boldsymbol{w} \rangle - \langle \nabla \boldsymbol{e}, \boldsymbol{u}\otimes \boldsymbol{w} \otimes \boldsymbol{v} \rangle,
\end{equation}
where $\boldsymbol{u} \in \mathbb{R}^3$, $\boldsymbol{v} \in \mathbb{R}^3$, and $\boldsymbol{w} \in \mathbb{R}^3$ are fixed, can be rewritten as
\begin{equation}
\llbracket \langle \nabla \boldsymbol{e}, \boldsymbol{f}_\alpha\otimes \boldsymbol{n} \otimes \boldsymbol{f}_\beta \rangle+\langle \nabla \boldsymbol{e}, \boldsymbol{f}_\beta\otimes \boldsymbol{n} \otimes \boldsymbol{f}_\alpha \rangle - \langle \nabla \boldsymbol{e}, \boldsymbol{f}_\alpha\otimes  \boldsymbol{f}_\beta \otimes \boldsymbol{n} \rangle - \langle \nabla_S \boldsymbol{n}, \boldsymbol{f}_\alpha\otimes  \boldsymbol{f}_\beta\rangle \langle \boldsymbol{e}, \boldsymbol{n}\otimes \boldsymbol{n} \rangle\rrbracket =0.
\end{equation}
The above equation, for different values of $\alpha$ and$\beta$, yields
\begin{eqnarray}
 \frac{1}{\sqrt{f_{11}}} \left\llbracket\frac{\partial \epsilon_{13}}{\partial \theta_1}\right\rrbracket-\frac{1}{\sqrt{f_{33}}} \left\llbracket\frac{\partial \epsilon_{11}}{\partial \theta_3}\right\rrbracket+\frac{1}{\sqrt{f_{11}f_{22}}} \frac{\partial \sqrt{f_{11}}}{\partial \theta_2} \llbracket\epsilon_{23}\rrbracket+ \frac{1}{2{f_{33}} \sqrt{f_{11}}} \frac{\partial {f_{33}}}{\partial \theta_1} \llbracket\epsilon_{13}\rrbracket -k_1 \llbracket\epsilon_{33}\rrbracket =0,\\
\frac{1}{\sqrt{f_{22}}} \left\llbracket\frac{\partial \epsilon_{23}}{\partial \theta_2}\right\rrbracket-\frac{1}{\sqrt{f_{33}}} \left\llbracket\frac{\partial \epsilon_{22}}{\partial \theta_3}\right\rrbracket+\frac{1}{\sqrt{f_{11}f_{22}}} \frac{\partial \sqrt{f_{22}}}{\partial \theta_1} \llbracket\epsilon_{13}\rrbracket+ \frac{1}{2{f_{33}} \sqrt{f_{22}}} \frac{\partial \sqrt{f_{33}}}{\partial \theta_2} \llbracket\epsilon_{23}\rrbracket -k_2 \llbracket\epsilon_{33}\rrbracket =0,
\\
 \frac{1}{\sqrt{f_{33}}} \left \llbracket\frac{\partial \epsilon_{12}}{\partial \theta_3}\right \rrbracket-\frac{1}{\sqrt{f_{22}}} \left \llbracket\frac{\partial \epsilon_{13}}{\partial \theta_2}\right \rrbracket-\frac{1}{\sqrt{f_{11}}} \left \llbracket\frac{\partial \epsilon_{23}}{\partial \theta_1}\right \rrbracket +\left( \frac{1}{\sqrt{f_{11}f_{22}}} \frac{\partial \sqrt{f_{11}}}{\partial \theta_2}  - \frac{1}{\sqrt{f_{22}f_{33}}} \frac{\partial \sqrt{f_{33}}}{\partial \theta_2} \right)  \llbracket\epsilon_{13}  \rrbracket \nonumber\\ 
  + \left( \frac{1}{\sqrt{f_{11} f_{22}}} \frac{\partial \sqrt{f_{22}}}{\partial \theta_1}- \frac{1}{\sqrt{f_{11}f_{33}}} \frac{\partial \sqrt{f_{33}}}{\partial \theta_1} \right) \llbracket\epsilon_{23}\rrbracket=0.
\end{eqnarray}
The interfacial compatibility conditions for a piecewise continuous strain field have been obtained in this form by Wheeler and Luo \cite{wheeler1999conditions} by considering the continuity of tangential strain and curvature across the interface. We note that the discontinuity in surface derivative of a field is same as the surface derivative of the discontinuity in the field, for instance $\llbracket {\partial \epsilon_{13}}/{\partial \theta_2} \rrbracket= {\partial \llbracket\epsilon_{13}\rrbracket}/{\partial \theta_2} $. This is however not the case with the discontinuity in normal derivative of a field. 
\end{rem}

\begin{rem} (Jump conditions in a curvilinear coordinate system) Let $(y_1, y_2, y_3) \in \mathbb{R}^3$ be a local parametrization of neighborhood of $S$ such that $S$ is given by $y_3 = 0$. The position vector in such neighborhoods can be written as $\boldsymbol{x}(y_1, y_2, y_3)= \boldsymbol{x}(y_1, y_2, 0) + y_3 \boldsymbol{n}$. The curvilinear covariant basis is defined by $\boldsymbol{g}_i={\partial \boldsymbol{x}}/{\partial y_i}$. The contravariant basis, $\boldsymbol{g}^i$, is defined by $\langle \boldsymbol{g}^i,\boldsymbol{g}_j \rangle=\delta^i_j$. Clearly, both $(\boldsymbol{g}_1,\boldsymbol{g}_2)$ and $(\boldsymbol{g}^1,\boldsymbol{g}^2)$, evaluated at $y_3=0$, can form a basis of the tangent plane on $S$. Also, $\boldsymbol{g}_3=\boldsymbol{g}^3=\boldsymbol{n}$ for $y_3=0$. The Christoffel symbols induced henceforth are given by $\Gamma^k_{ij}=\langle {\partial \boldsymbol{g}_i}/{\partial y_j},\boldsymbol{g}^k\rangle$. Moreover, we choose the parametrization such that $\boldsymbol{g}_1\times \boldsymbol{g}_2=|\boldsymbol{g}_1\times \boldsymbol{g}_2|\boldsymbol{n}$, $\boldsymbol{n}\times \boldsymbol{g}^1=({|\boldsymbol{g}^1|}/{|\boldsymbol{g}_2|}) \boldsymbol{g}_2$, and $\boldsymbol{n}\times \boldsymbol{g}^2=-({|\boldsymbol{g}^2|}/{|\boldsymbol{g}_1|}) \boldsymbol{g}_1$.
Let $h_{ij}$ be the covariant components of the strain field $\boldsymbol{e}$ with respect to the defined covariant basis, i.e., we can write  $\boldsymbol{e}= h_{ij}(\boldsymbol{g}^i\otimes \boldsymbol{g}^j)$ in the vicinity of $S$. We have ${\partial \boldsymbol{e}}/{\partial y_k}=h_{ij || k}(\boldsymbol{g}^i\otimes \boldsymbol{g}^j)$, where 
$h_{ij || k} = {\partial h_{ij}}/{\partial y_k} - \Gamma^l_{ki}h_{lj} - \Gamma^l_{kj}h_{il}$ is the covariant derivative.
The jump condition \eqref{I2} is equivalent to 
$( \llbracket \curl \boldsymbol{e} \rrbracket  \times \boldsymbol{n} )^T+\llbracket \curl_S \left(( \boldsymbol{e} \times \boldsymbol{n})^T\right)\rrbracket, \boldsymbol{g}^\beta\otimes \boldsymbol{g}^\alpha \rangle=0$ for all $\alpha$, $\beta$, which on using Equation \eqref{curl_curvilinear} takes the form
\begin{equation}
 \langle \llbracket \nabla \boldsymbol{e}\rrbracket, \boldsymbol{g}_\alpha\otimes \boldsymbol{n} \otimes \boldsymbol{g}_\beta \rangle+\langle \llbracket \nabla \boldsymbol{e}\rrbracket, \boldsymbol{g}_\beta\otimes \boldsymbol{n} \otimes \boldsymbol{g}_\alpha \rangle - \langle \llbracket \nabla \boldsymbol{e}\rrbracket, \boldsymbol{g}_\alpha\otimes  \boldsymbol{g}_\beta \otimes \boldsymbol{n} \rangle - \langle \nabla_S \boldsymbol{n}, \boldsymbol{g}_\alpha\otimes  \boldsymbol{g}_\beta\rangle  \langle \llbracket \boldsymbol{e}\rrbracket, \boldsymbol{n}\otimes \boldsymbol{n} \rangle  =0.
\end{equation}
The interfacial compatibility conditions \eqref{I1} and \eqref{I2}, consequently, can be written as
\begin{equation}
\label{I2Ciarlet}
 \llbracket h_{\alpha \beta}\rrbracket =0 ~\text{and}~ \llbracket h_{\alpha 3 || \beta}\rrbracket + \llbracket h_{\beta 3 || \alpha}\rrbracket -  \llbracket h_{\alpha \beta || 3}\rrbracket + \Gamma^3_{\alpha\beta}  \llbracket h_{33}\rrbracket  =0,
\end{equation}
respectively. 
\label{curvijump}
\end{rem}

\begin{rem} (Compatibility conditions for displacement boundary conditions)
 We call a smooth strain field $\boldsymbol{e}$ in $\Omega$ to be compatible with the displacement boundary condition if and only if there exists a smooth vector valued field $\boldsymbol{u}$ in $\Omega$ such that $\boldsymbol{u}|_{\partial \Omega_1}=\boldsymbol{0}$ and $\boldsymbol{e}=(1/2)({\nabla \boldsymbol{u} +\boldsymbol{u}^T})$, where $\partial \Omega_1$ is the part of the boundary $\partial \Omega$ where displacement field is specified. Towards this end, we consider domain $\Omega$ to be contained within a larger domain ${\Omega}_l \subset \mathbb{R}^3$ such that   $\partial \Omega_1= \partial \Omega \cap \partial (\Omega_l - \Omega)$. Clearly, the trivial strain field $\boldsymbol{e} =\boldsymbol{0}$ in $\Omega_l - \Omega$ is compatible with the boundary condition $\boldsymbol{u}=\boldsymbol{0}$ on $\partial \Omega_1$. We consider a symmetric tensor valued distribution $\boldsymbol{E}\in \mathcal{B}(\Omega_l ,\Sym)$ with bulk density $\boldsymbol{e}$ in $\Omega$ and $\boldsymbol{0}$ in $\Omega_l - \Omega$. The compatibility of $\boldsymbol{e}$ with $\boldsymbol{u}|_{\partial \Omega_1}=\boldsymbol{0}$ is then ensured by relation \eqref{BulkCompatibility} in $\Omega$ and the following boundary conditions, as deduced from Equations \eqref{I1} and \eqref{I2},
 \begin{equation}
\left(( \boldsymbol{e} \times \boldsymbol{n})^T\times \boldsymbol{n}\right)^T=\boldsymbol{0} ~\text{on}~\partial \Omega_1 ~\text{and}
\end{equation}
\begin{equation}
( \curl \boldsymbol{e}  \times \boldsymbol{n} )^T+\curl_S ( \boldsymbol{e} \times \boldsymbol{n})^T=\boldsymbol{0}~\text{on}~\partial \Omega_1.
\end{equation}
The above represent conditions on strain which are equivalent to imposing homogeneous displacement boundary condition on some part of the boundary. We will consider the conditions for heterogeneous displacement boundary condition in Remark \ref{inhomobc}.
In terms of the curvilinear coordinate system, as introduced in Remark \ref{curvijump}, the interfacial conditions become
\begin{equation}
h_{\alpha \beta} =0~\text{and}~h_{\alpha 3 || \beta} +  h_{\beta 3 || \alpha} -   h_{\alpha \beta || 3} + \Gamma^3_{\alpha\beta}   h_{33}  =0.
\end{equation}
These relations have been previously obtained by Ciarlet and Mardare \cite{ciarlet2014intrinsic} by considering the linearized form of the first and second fundamental forms induced by the strain on the boundary. That these boundary conditions can be obtained for strain tensor belonging to weaker functional spaces has also been established in the same paper. 
\label{ccdbc}
\end{rem}

\subsection{Imperfectly Bonded Surface of Discontinuity}
\label{cs}

Let $\boldsymbol{e}_B$ be a piecewise smooth symmetric tensor field on a simply connected domain $\Omega$, possibly discontinuous across a regular oriented surface $S \in \Omega$ with $\partial S - \partial \Omega = \emptyset$, and let $\boldsymbol{e}_S$ be a smooth symmetric tensor field on $S$. Then, for a compactly supported smooth tensor valued field $\boldsymbol{\phi} \in \mathcal{D}({\Omega, \Lin})$, we can define a distribution $\boldsymbol{E} \in \mathcal{B}(\Omega,\Sym)$ such that
\begin{equation}
\boldsymbol{E}(\boldsymbol{\phi})= \int_{\Omega} \langle \boldsymbol{e}_B,\boldsymbol{\phi} \rangle dv + \int_{S} \langle \boldsymbol{e}_S,\boldsymbol{\phi} \rangle da. \label{diststraincs}
\end{equation}
Clearly, $\boldsymbol{E}$ is composed of distributions $\boldsymbol{E}_B \in \mathcal{B}(\Omega,\Sym)$ and $\boldsymbol{E}_S  \in \mathcal{C}(\Omega,\Sym)$ such that $\boldsymbol{E}_B(\boldsymbol{\phi})=\int_{\Omega} \langle\boldsymbol{e}_B,\boldsymbol{\phi} \rangle dv$ and $\boldsymbol{E}_S(\boldsymbol{\phi})=\int_S \langle \boldsymbol{e}_S, \boldsymbol{\phi} \rangle  da$. Using the results from the beginning of Section \ref{pcs}, we can write
\begin{eqnarray}
\Curl \Curl \boldsymbol{E}_B (\boldsymbol{\phi})= \int_{\Omega} \langle \curl \curl \boldsymbol{e}_B, \boldsymbol{\phi} \rangle dv + \int_S \Big( \Big\langle \big( (\llbracket \curl \boldsymbol{e}_B \rrbracket \times \boldsymbol{n} )^T -\kappa \left((\llbracket \boldsymbol{e}_B \rrbracket\times \boldsymbol{n})^T\times \boldsymbol{n}\right)^T \nonumber \\ +\curl_S (\llbracket \boldsymbol{e}_B \rrbracket\times \boldsymbol{n})^T \big),\boldsymbol{\phi} \Big\rangle+\Big\langle \left((\llbracket \boldsymbol{e}_B \rrbracket\times \boldsymbol{n})^T\times \boldsymbol{n}\right)^T,\frac{\partial \boldsymbol{\phi}}{\partial n} \Big\rangle \Big)da. 
\end{eqnarray}
On other other hand, Identity \eqref{CurlS} implies
\begin{equation}
\Curl \boldsymbol{E}_S (\boldsymbol{\phi})=\int_S  \left\langle -\kappa (\boldsymbol{e}_S\times \boldsymbol{n})^T+ \curl_S \boldsymbol{e}_S,\boldsymbol{\phi} \right\rangle da + \int_S \left\langle (\boldsymbol{e}_S \times \boldsymbol{n})^T ,\frac{\partial \boldsymbol{\phi}}{\partial n}   \right\rangle da,
\end{equation}
which, on using Identities \eqref{CurlS} and \eqref{CurlF}, yields $\Curl \Curl \boldsymbol{E}_S (\boldsymbol{\phi})=$
\begin{eqnarray}
\int_S  \Big\langle \Big( \kappa^2 \left((\boldsymbol{e}_S\times \boldsymbol{n})^T \times \boldsymbol{n} \right)^T- \kappa (\curl_S \boldsymbol{e}_S\times \boldsymbol{n})^T  - \curl_S (\kappa (\boldsymbol{e}_S\times \boldsymbol{n})^T)+ \curl_S \curl_S \boldsymbol{e}_S\nonumber \\
 - \div_S \left(\nabla_S \boldsymbol{n} \times (\boldsymbol{e}_S\times \boldsymbol{n})^T \right) \Big),\boldsymbol{\phi} \Big\rangle da  + \int_S \Big\langle \Big( -2\kappa \left((\boldsymbol{e}_S \times \boldsymbol{n})^T \times \boldsymbol{n}\right)^T  +(\curl_S \boldsymbol{e}_S\times \boldsymbol{n})^T  \nonumber \\
  + \curl_S (\boldsymbol{e}_S\times \boldsymbol{n})^T \Big),\frac{\partial \boldsymbol{\phi}}{\partial n} \Big\rangle da   +\int_S \left\langle \left((\boldsymbol{e}_S\times \boldsymbol{n})^T \times \boldsymbol{n} \right)^T, \left(\nabla(\nabla \boldsymbol{\phi})\boldsymbol{n}\otimes\boldsymbol{n}\right)    \right\rangle da.
\end{eqnarray}
The condition $\Curl \Curl \boldsymbol{E} (\boldsymbol{\phi})=\boldsymbol{0}$, for arbitrary $\boldsymbol{\phi}$, is therefore equivalent to requiring
 \begin{eqnarray}
 \curl \curl \boldsymbol{e}_B=\boldsymbol{0}~\text{in}~\Omega - S,  \label{BulkCompatibilityC} \\
\left(( \boldsymbol{e}_S \times \boldsymbol{n})^T\times \boldsymbol{n}\right)^T=\boldsymbol{0}~\text{on}~S, \label{I1C} \\
(\curl_S \boldsymbol{e}_S \times \boldsymbol{n})^T  + \curl_S (\boldsymbol{e}_S\times \boldsymbol{n})^T + \left((\llbracket \boldsymbol{e}_B \rrbracket\times \boldsymbol{n})^T\times \boldsymbol{n}\right)^T=\boldsymbol{0}~\text{on}~S,~\text{and} \label{I2C}
\\
(\llbracket \curl \boldsymbol{e}_B \rrbracket \times \boldsymbol{n} )^T+\curl_S (\llbracket \boldsymbol{e}_B \rrbracket\times \boldsymbol{n})^T + ((\boldsymbol{e}_S\times \boldsymbol{n})^T\times \nabla_S \kappa)^T \nonumber \\ + \curl_S \curl_S \boldsymbol{e}_S - \div_S \left(\nabla_S \boldsymbol{n} \times (\boldsymbol{e}_S\times \boldsymbol{n})^T\right)=\boldsymbol{0}~\text{on}~S, \label{I3C}
\end{eqnarray}
where the identity $\curl_S (\kappa \boldsymbol{e})=\kappa \curl_S\boldsymbol{e}- (\boldsymbol{e}\times \nabla_S \kappa)^T$ has been used to obtain Equation \eqref{I3C}.
On the other hand, according to Corollary \ref{StrainCompatibilityBC}, $\Curl \Curl \boldsymbol{E}=  \boldsymbol{0}$, with $\boldsymbol{E}$ given by \eqref{diststraincs}, is equivalent to existence of a $\boldsymbol{U} \in \mathcal{B}(\Omega, \mathbb{R}^3)$ such that $\boldsymbol{E}= (1/2)({\nabla \boldsymbol{U} +(\nabla \boldsymbol{U})^T})$, with $\boldsymbol{U}(\boldsymbol{\psi})=\int_\Omega \langle \boldsymbol{u},\boldsymbol{\psi}\rangle dv$, for $\boldsymbol{\psi}\in \mathcal{D}(\Omega,\mathbb{R}^3)$, where $\boldsymbol{u}$ is a piecewise smooth vector field on $\Omega$, possibly discontinuous across $S$. 
Summarizing the above, we have
\begin{proposition}
For a piecewise smooth tensor valued field $\boldsymbol{e}_B$ on a simply connected domain $\Omega \subset \mathbb{R}^3$, allowed to be discontinuous across an oriented regular surface $S \subset \Omega$ with unit normal $\boldsymbol{n}$ and  $\partial S - \partial \Omega = \emptyset$, and a smooth tensor valued field $\boldsymbol{e}_S$ on $S$, there exists a piecewise smooth vector valued field $\boldsymbol{u}$ on $\Omega$  such that $\boldsymbol{e}_B=(1/2)(\nabla \boldsymbol{u}+(\nabla \boldsymbol{u})^T)$ in $\Omega - S$ and $\boldsymbol{e}_S=-(1/2) (\llbracket \boldsymbol{u} \rrbracket\otimes \boldsymbol{n} + \boldsymbol{n} \otimes \llbracket \boldsymbol{u} \rrbracket)$ on $S$ if and only if $\boldsymbol{e}_B$ and $\boldsymbol{e}_S$ satisfy Equations \eqref{BulkCompatibilityC}, \eqref{I1C}, \eqref{I2C}, and \eqref{I3C}.
\end{proposition}

\begin{rem} (Planar strain field) As an immediate application of the preceding compatibility equations, we recall the planar strain field case, as discussed in Remark \ref{psf}, and seek the conditions on bulk strain such that there exist a displacement field $\boldsymbol{u}$ which satisfies $\boldsymbol{e}_B=(1/2)(\nabla \boldsymbol{u}+(\nabla \boldsymbol{u})^T)$ in $\Omega - S$ and $\langle \llbracket \boldsymbol{u} \rrbracket, \boldsymbol{n} \rangle = 0$ on $S$. We use the same notation as in Remark \ref{psf}. Consider $\boldsymbol{e}_S$ such that $\langle \boldsymbol{e}_S,\boldsymbol{n}\otimes \boldsymbol{n} \rangle=0$. This, along with Equation \eqref{I1C}, implies that $\boldsymbol{e}_S$ is of the form $\boldsymbol{e}_S=a (\boldsymbol{t}\otimes  \boldsymbol{n} +\boldsymbol{n}\otimes  \boldsymbol{t})$, where $a$ is a smooth scalar field on $S$. Consequently, Equation \eqref{I2C}, on recalling the plane strain assumption, reduces to $2a' +  \left\llbracket e_{ij} \right\rrbracket t_i t_j = 0$, where the superscript prime denotes a derivative along the curve $C$. Moreover, the three terms in Equation \eqref{I3C} involving $\boldsymbol{e}_S$ can be simplified to $2k'a + 4 k a'$.  We can then eliminate $a$ between Equations \eqref{I2C} and \eqref{I3C} to obtain the following condition on $\boldsymbol{e}_B$ across $C$:
\begin{equation}
\left\llbracket e_{ij} \right\rrbracket t_i t_j=\left(\frac{1}{k'}\left( \left\llbracket \frac{\partial e_{ij}}{\partial x_p} \right\rrbracket n_i t_j t_p + \left\llbracket \left(\frac{\partial e_{ij}}{\partial x_p} - \frac{\partial e_{pj}}{\partial x_i}\right) \right\rrbracket n_i t_j t_p + k\left\llbracket e_{ij} \right\rrbracket n_i n_j -2k \left\llbracket e_{ij} \right\rrbracket t_i t_j \right)\right)' \label{psfimcc}
\end{equation}
whenever $k' \neq 0$ and
\begin{equation}
\left( \left\llbracket \frac{\partial e_{ij}}{\partial x_p} \right\rrbracket n_i t_j t_p + \left\llbracket \left(\frac{\partial e_{ij}}{\partial x_p} - \frac{\partial e_{pj}}{\partial x_i}\right) \right\rrbracket n_i t_j t_p + k\left\llbracket e_{ij} \right\rrbracket n_i n_j -2k \left\llbracket e_{ij} \right\rrbracket t_i t_j \right)=0
\end{equation}
when $k' = 0$. These are the required conditions on the bulk strain field. The condition \eqref{psfimcc} has been previously obtained by Markenscoff \cite{markenscoff1996note}.
 We can also view these interfacial conditions as those required on $\boldsymbol{e}_B$ such that there exists a concentrated slip strain $\boldsymbol{e}_S$ on $S$, with $\langle \boldsymbol{e}_S,\boldsymbol{n}\otimes \boldsymbol{n} \rangle=0$, for which  $\Curl \Curl \boldsymbol{E} = \boldsymbol{0}$.
\end{rem}

\begin{rem}
\label{inhomobc}
(Heterogeneous boundary conditions for displacement)  In Remark \ref{ccdbc}, we discussed the compatibility of a bulk strain field $\boldsymbol{e}$ with homogeneous displacement boundary conditions. We will now extend that result to include heterogeneous boundary conditions $\boldsymbol{u}|_{\partial \Omega_1}= \hat{\boldsymbol{u}}$, where $\hat{\boldsymbol{u}} \in C^{\infty} (\partial \Omega_1, \mathbb{R}^3)$. For the domain $\Omega_l$, as introduced in Remark \ref{ccdbc}, we consider $\boldsymbol{E}\in \mathcal{D}'(\Omega_l,\Sym)$ such that $\boldsymbol{E} =\boldsymbol{E}_1+\boldsymbol{E}_2$, where $\boldsymbol{E}_1 \in \mathcal{B}(\Omega_l,\Sym)$ and $\boldsymbol{E}_2 \in \mathcal{C}(\Omega_l,\Sym)$. The bulk density field, used to construct $\boldsymbol{E}_1$, is taken as $\boldsymbol{e}_B = \boldsymbol{e}$ in $\Omega$ and $\boldsymbol{0}$ otherwise. The surface density field for constructing $\boldsymbol{E}_2$ is taken as $\boldsymbol{e}_S= -({1}/{2})(\hat{\boldsymbol{u}}\otimes \boldsymbol{n}+\boldsymbol{n}\otimes \hat{\boldsymbol{u}})$ on $\partial \Omega_1$. The compatibility of $\boldsymbol{e}$ with $\boldsymbol{u}|_{\partial \Omega_1}=\hat{\boldsymbol{u}}$ is then ensured by relation \eqref{BulkCompatibility} in $\Omega$ and the following boundary conditions, as deduced from Equations \eqref{I2C} and \eqref{I3C},
\begin{eqnarray}
(\curl_S \boldsymbol{e}_S\times \boldsymbol{n})^T  + \curl_S (\boldsymbol{e}_S\times \boldsymbol{n})^T + (( \boldsymbol{e} \times \boldsymbol{n})^T\times \boldsymbol{n})^T=\boldsymbol{0}~\text{on}~\partial \Omega_1~\text{and}
\\
( \curl \boldsymbol{e}  \times \boldsymbol{n} )^T+\curl_S ( \boldsymbol{e} \times \boldsymbol{n})^T + ((\boldsymbol{e}_S\times \boldsymbol{n})^T\times \nabla_S \kappa)^T \nonumber
\\  + \curl_S \curl_S \boldsymbol{e}_S - \div_S (\nabla_S \boldsymbol{n} \times (\boldsymbol{e}_S\times \boldsymbol{n})^T)=\boldsymbol{0}~\text{on}~\partial \Omega_1,
\end{eqnarray}  
where $\boldsymbol{e}_S= -({1}/{2})(\hat{\boldsymbol{u}}\otimes \boldsymbol{n}+\boldsymbol{n}\otimes \hat{\boldsymbol{u}})$ is known. The compatibility condition \eqref{I1C} is trivially satisfied for the form of $\boldsymbol{e}_S$ considered here.
In terms of the curvilinear coordinate system, as introduced in Remark \ref{curvijump},  the above interfacial conditions reduce to
\begin{eqnarray}
e_{\alpha\beta}=(1/2)(\langle\partial_\alpha \hat{\boldsymbol{u}},\boldsymbol{g}_\beta\rangle+\langle\partial_\beta \hat{\boldsymbol{u}},\boldsymbol{g}_\alpha\rangle)~\text{on}~\partial \Omega_1~\text{and}
\\
e_{\alpha 3 || \beta} +  e_{\beta 3 || \alpha} -   e_{\alpha \beta || 3} + \Gamma^3_{\alpha\beta}   e_{33}  =\langle (\partial_{\alpha \beta}\hat{\boldsymbol{u}} -\Gamma^{\sigma}_{\alpha \beta}\partial_{\sigma}\hat{\boldsymbol{u}}),\boldsymbol{n}\rangle~\text{on}~\partial \Omega_1.
\end{eqnarray}
These relations in the above form have been obtained by Ciarlet and Mardare \cite{ciarlet2014intrinsic}.
\end{rem}

\section{Topological Defects and Metric Anomalies as Sources of Incompatibility}
\label{inhomoincomp}

It is well known that the presence of defects and metric anomalies is related to incompatibility of strain field \cite{kroner81, de1981view} and consequently to being sources of internal stress field. In the following we consider dislocations, disclinations, and metric anomalies in the form of piecewise smooth bulk densities, smooth surface densities, and smooth surface densities of defect dipoles. Using the theory of distributions, we relate these defect densities to kinematical quantities given by strain and bend-twist fields thereby generalizing the expressions derived earlier by de Wit \cite{de1981view}, where the formulation was restricted to smooth bulk fields. This leads us to the main result of the paper, that is to express strain incompatibility in terms of the introduced defect densities, both on the interface and away from it.  We provide several remarks including those related to defect conservation laws, dislocation loops, plane strain simplification, and nilpotent defect densities.

\subsection{Defects as Distributions and their Relationship with Strains}

Given a piecewise smooth dislocation density tensor field $\boldsymbol{\alpha}_B$ over $\Omega-S$, possibly discontinuous across $S$ with $S$ such that $\partial S - \partial \Omega =\emptyset$, and smooth dislocation density tensor fields $\boldsymbol{\alpha}_{S_1}$ and $\boldsymbol{\alpha}_{S_2}$ on $S$, we can introduce distributions $\boldsymbol{A}_B\in \mathcal{B}(\Omega,\Lin)$, $\boldsymbol{A}_1\in \mathcal{C}(\Omega,\Lin)$, and $\boldsymbol{A}_2\in \mathcal{F}(\Omega,\Lin)$ such that, for $\boldsymbol{\phi}\in \mathcal{D}(\Omega,\Lin)$,
\begin{equation}
\boldsymbol{A}_B (\boldsymbol{\phi})= \int_\Omega \langle \boldsymbol{\alpha}_B,\boldsymbol{\phi} \rangle dv,~\boldsymbol{A}_1 (\boldsymbol{\phi})= \int_S \langle \boldsymbol{\alpha}_{S_1},\boldsymbol{\phi} \rangle da,~\text{and}~ \boldsymbol{A}_2 (\boldsymbol{\phi})=\int_S \left\langle \boldsymbol{\alpha}_{S_2},\frac{\partial \boldsymbol{\phi}}{\partial n} \right\rangle da. \label{dddef}
\end{equation}
Whereas the notions of $\boldsymbol{\alpha}_B$, as a bulk dislocation density, and $\boldsymbol{\alpha}_{S_1}$, as a surface dislocation density, are well established in the literature \cite{kroner81, bilby}, the latter being used, e.g., to represent dislocation walls, the meaning of surface density $\boldsymbol{\alpha}_{S_2}$ requires some further discussion. As we shall argue,  it represents a surface density of dislocation couples. Using the definitions \eqref{dddef} we can introduce a distribution $\boldsymbol{A} \in \mathcal{D}'(\Omega,\Lin)$ such that $\boldsymbol{A} = \boldsymbol{A}_B + \boldsymbol{A}_1 + \boldsymbol{A}_2$, i.e.,
\begin{equation}
\label{distA}
\boldsymbol{A}(\boldsymbol{\phi})=\int_\Omega \langle \boldsymbol{\alpha}_B,\boldsymbol{\phi} \rangle dv +\int_S \langle \boldsymbol{\alpha}_{S_1},\boldsymbol{\phi} \rangle da+\int_S \left\langle \boldsymbol{\alpha}_{S_2},\frac{\partial \boldsymbol{\phi}}{\partial n} \right\rangle da.
 \end{equation}
In terms of the above dislocation density fields, we can define the corresponding contortion tensors as $\boldsymbol{\gamma}_B=\boldsymbol{\alpha}_B-(1/2)(\tr \boldsymbol{\alpha}_B) \boldsymbol{I}$, $\boldsymbol{\gamma}_{S_1}=\boldsymbol{\alpha}_{S_1}-({1}/{2})(\tr \boldsymbol{\alpha}_{S_1})\boldsymbol{I}$, and $\boldsymbol{\gamma}_{S_2}=\boldsymbol{\alpha}_{S_2}-({1}/{2}) (\tr \boldsymbol{\alpha}_{S_2})\boldsymbol{I}$, so as to subsequently introduce a distribution $\boldsymbol{\Gamma} \in \mathcal{D}'(\Omega,\Lin)$ such that 
\begin{equation}
\boldsymbol{\Gamma}(\boldsymbol{\phi})=\int_\Omega \langle \boldsymbol{\gamma}_B,\boldsymbol{\phi} \rangle dv +\int_S \langle \boldsymbol{\gamma}_{S_1},\boldsymbol{\phi} \rangle da+\int_S \left\langle \boldsymbol{\gamma}_{S_2},\frac{\partial \boldsymbol{\phi}}{\partial n} \right\rangle da. \label{distgamma}
 \end{equation}

\begin{figure}[t!]
 \centering
 \centering{\includegraphics[scale=0.6, angle=0]{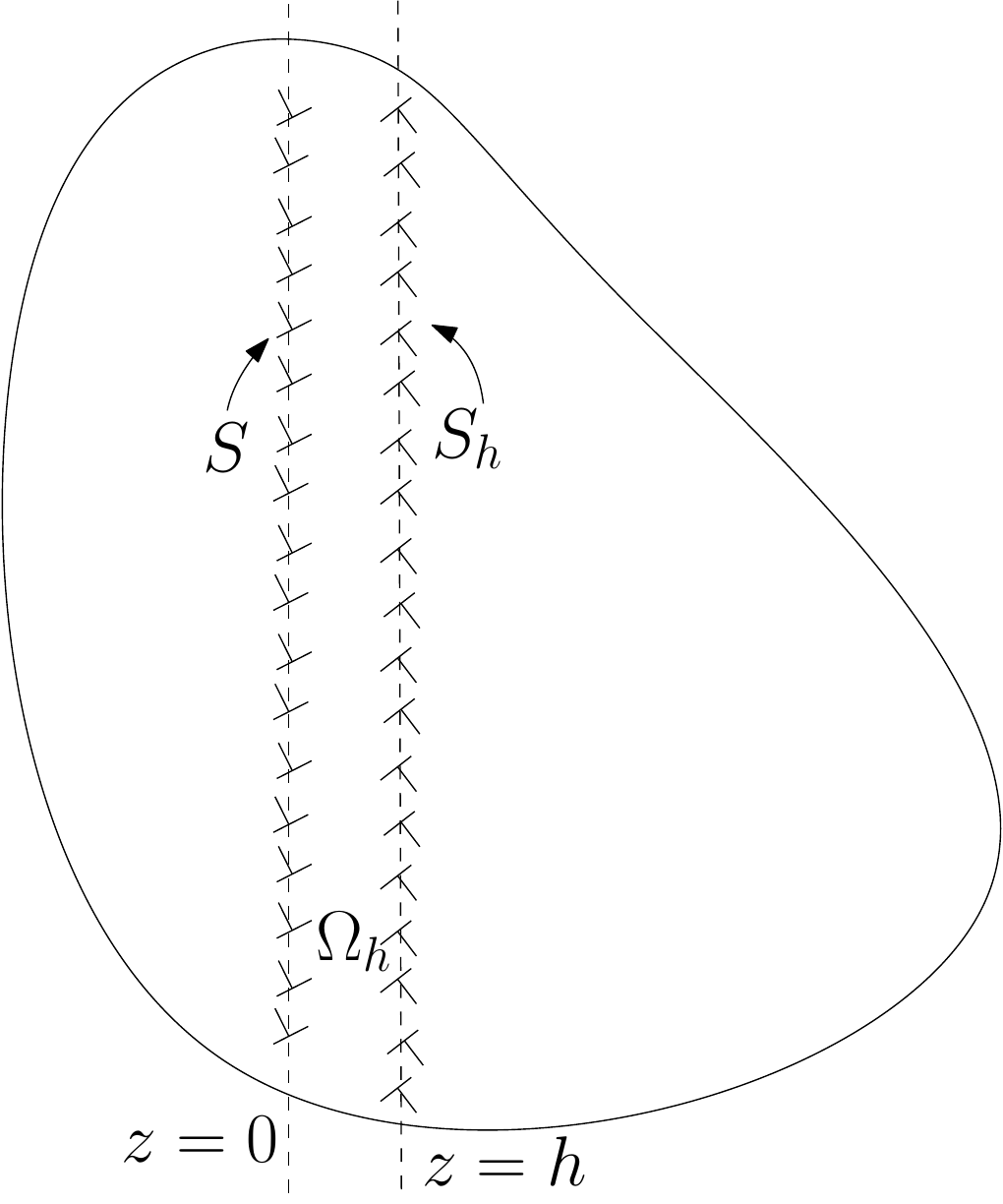}}
\caption{A pair of dislocation walls with equal and opposite charge.}
\label{dislwall}
 \end{figure}
 
To understand the significance of $\boldsymbol{A}_2$, and the associated density $\boldsymbol{\alpha}_{S_2}$, we consider two mutually parallel plane surfaces $S$, with normal $\boldsymbol{e}_3$ given by $z=0$, and $S_h$, given by $z=h$. The bulk region enclosed by the two surfaces ($0<z<h$) is denoted by $\Omega_h$. Let $\boldsymbol{A}_h \in \mathcal{D}'(\Omega,\Lin)$ be such that, for any $\boldsymbol{\phi}\in \mathcal{D}(\Omega,\Lin)$, 
\begin{equation}
\boldsymbol{A}_h (\boldsymbol{\phi}) = -\int_S \left\langle \frac{\boldsymbol{\alpha}_0}{h} , \boldsymbol{\phi} \right\rangle da + \int_{S_h} \left\langle \frac{\boldsymbol{\alpha}_0}{h} , \boldsymbol{\phi} \right\rangle da, \label{ddipoles}
\end{equation} 
where $\boldsymbol{\alpha}_0 \in \Lin$ is a constant. The two integrands represent dislocation walls, separated by a distance $h$, with uniform density of dislocations but with opposite sign. The surface densities are uniform and scale as the inverse of the distance between walls. For infinitesimal distance between the dislocation walls ($h \to 0$), $\boldsymbol{\alpha}_h (\boldsymbol{\phi}) \to \boldsymbol{A}_0 (\boldsymbol{\phi})$, with $\boldsymbol{A}_0 \in \mathcal{F}(\Omega,\Lin)$, where $\boldsymbol{A}_0 (\boldsymbol{\phi}) = \int_S \langle \boldsymbol{\alpha}_0,{\partial \boldsymbol{\phi}}/{\partial n} \rangle da$.
Therefore $\boldsymbol{A}_0 \in \mathcal{F}(\Omega,\Lin)$, with planar surface and uniform surface density, can be interpreted in terms of  two dislocation walls, infinitesimally close to each other, and with surface densities of opposite sign scaling as the inverse of the distance between the walls. A pair of dislocation walls, as discussed here, is illustrated in Figure \ref{dislwall}.

 In an analogous manner, given a piecewise smooth disclination density tensor field $\boldsymbol{\theta}_B$ over $\Omega-S$, possibly discontinuous across $S$ with $S$ such that $\partial S - \partial \Omega =\emptyset$, and smooth disclination density tensor fields $\boldsymbol{\theta}_{S_1}$ and $\boldsymbol{\theta}_{S_2}$ on $S$, we can introduce distributions $\boldsymbol{\Theta}_B\in \mathcal{B}(\Omega,\Lin)$, $\boldsymbol{\Theta}_1\in \mathcal{C}(\Omega,\Lin)$, and $\boldsymbol{\Theta}_2\in \mathcal{F}(\Omega,\Lin)$ such that, for $\boldsymbol{\phi}\in \mathcal{D}(\Omega,\Lin)$,
\begin{equation}
\boldsymbol{\Theta}_B (\boldsymbol{\phi})= \int_\Omega \langle \boldsymbol{\theta}_B,\boldsymbol{\phi} \rangle dv,~\boldsymbol{\Theta}_1 (\boldsymbol{\phi})= \int_S \langle \boldsymbol{\theta}_{S_1},\boldsymbol{\phi} \rangle da,~\text{and}~ \boldsymbol{\Theta}_2 (\boldsymbol{\phi})=\int_S \left\langle \boldsymbol{\theta}_{S_2},\frac{\partial \boldsymbol{\phi}}{\partial n} \right\rangle da. \label{diddef}
\end{equation}
Clearly, $\boldsymbol{\theta}_B$ represents a bulk disclination density field and $\boldsymbol{\theta}_{S_1}$ a density of disclinations spread over the surface $S$. Moreover, following an argument, similar to that mentioned in the preceding paragraph, we can interpret $\boldsymbol{\theta}_{S_2}$ as a surface distribution of disclination dipoles. Using the definitions \eqref{diddef} we can introduce a distribution $\boldsymbol{\Theta} \in \mathcal{D}'(\Omega,\Lin)$ 
such that $\boldsymbol{\Theta}=\boldsymbol{\Theta}_B+\boldsymbol{\Theta}_1+\boldsymbol{\Theta}_2$, i.e.,
\begin{equation}
\boldsymbol{\Theta}(\boldsymbol{\phi})=\int_\Omega \langle \boldsymbol{\theta}_B,\boldsymbol{\phi} \rangle dv +\int_S \langle \boldsymbol{\theta}_{S_1},\boldsymbol{\phi} \rangle da+\int_S \left\langle \boldsymbol{\theta}_{S_2},\frac{\partial \boldsymbol{\phi}}{\partial n} \right\rangle da.
 \end{equation} 

Besides dislocations and dislocations, we also include metric anomalies as possible sources of strain incompatibility. The metric anomalies, which can appear due to thermal strains, growth strains, extra-matter, interstitials, etc., are given by a piecewise smooth density symmetric tensor field $\boldsymbol{e}^Q_B$ over $\Omega-S$, possible discontinuous across $S$ with $S$ such that $\partial S - \partial \Omega =\emptyset$, and a smooth surface density symmetric tensor field  $\boldsymbol{e}^Q_S$ over $S$. We can introduce distributions $\boldsymbol{E}^Q_B \in \mathcal{B}(\Omega,\Sym)$ and $\boldsymbol{E}^Q_S \in \mathcal{C}(\Omega,\Sym)$ such that, for $\boldsymbol{\phi}\in \mathcal{D}(\Omega,\Lin)$,
\begin{equation}
\boldsymbol{E}^Q_B (\boldsymbol{\phi})=\int_{\Omega}\langle \boldsymbol{e}^Q_B,\boldsymbol{\phi}\rangle dv ~\text{and}~ \boldsymbol{E}^Q_S (\boldsymbol{\phi}) = \int_S\langle \boldsymbol{e}^Q_S,\boldsymbol{\phi}\rangle da.
\end{equation}
We can also introduce a distribution $\boldsymbol{E}^Q \in \mathcal{D}'(\Omega,\Sym)$ such that $\boldsymbol{E}^Q=\boldsymbol{E}^Q_B+\boldsymbol{E}^Q_S$, i.e.,  
\begin{equation}
\boldsymbol{E}^Q(\boldsymbol{\phi})=\int_{\Omega}\langle \boldsymbol{e}^Q_B,\boldsymbol{\phi}\rangle dv +\int_S\langle \boldsymbol{e}^Q_S,\boldsymbol{\phi}\rangle da. \label{distmetano}
\end{equation} 

The distributions $\boldsymbol{A}$, $\boldsymbol{\Theta}$, and $\boldsymbol{E}^Q$ contain all the prescribed information regarding various defect densities and metric anomalies over the body $\Omega$ and the surface $S$. We would, next, like to relate defect densities to kinematical fields. Towards this end, we introduce two distributions $\boldsymbol{E}\in \mathcal{B}(\Omega,\Sym)$ and $\boldsymbol{K}=\boldsymbol{K}_1+\boldsymbol{K}_2$, where $\boldsymbol{K}_1\in \mathcal{B}(\Omega,\Lin)$ and $\boldsymbol{K}_2\in \mathcal{C}(\Omega,\Lin)$, such that, for $\boldsymbol{\phi}\in \mathcal{D}(\Omega,\Lin)$,
\begin{equation}
\boldsymbol{E}(\boldsymbol{\phi})=\int_\Omega \langle \boldsymbol{e},\boldsymbol{\phi} \rangle dv ~\text{and}~\boldsymbol{K}(\boldsymbol{\phi})=\int_\Omega \langle \boldsymbol{\kappa}_B,\boldsymbol{\phi} \rangle dv +\int_S \langle \boldsymbol{\kappa}_S,\boldsymbol{\phi} \rangle da, \label{strbend}
 \end{equation}
with $S$ such that $\partial S - \partial \Omega =\emptyset$, where $\boldsymbol{e}$ is the piecewise smooth strain field over $\Omega-S$, possibly discontinuous across $S$, $\boldsymbol{\kappa}_B$ is the piecewise smooth bend-twist field over $\Omega-S$ \cite{kroner81, de1981view}, possibly discontinuous across $S$, and $\boldsymbol{\kappa}_S$ is the smooth surface bend-twist field over $S$.

Drawing an analogy from the classical framework of de Wit \cite{de1981view}, where only smooth defect densities and kinematic fields were considered, we postulate the following relationships between the above defined distributions:
 \begin{eqnarray}
 \boldsymbol{\Theta}=\Curl \boldsymbol{K}^T~\text{and} \label{distdefeq1}
\\
 \boldsymbol{A} = \Curl (\boldsymbol{E}-\boldsymbol{E}^Q) + \tr(\boldsymbol{K})\boldsymbol{I} - \boldsymbol{K}^T. \label{distdefeq2}
\end{eqnarray}   
In the absence of defects, the above equations imply (for a simply connected $\Omega$) the existence of a $\boldsymbol{U} \in \mathcal{B}(\Omega, \mathbb{R}^3)$ such that $\boldsymbol{E}= (1/2)({\nabla \boldsymbol{U} +(\nabla \boldsymbol{U})^T})$, with $\boldsymbol{U}(\boldsymbol{\psi})=\int_\Omega \langle \boldsymbol{u},\boldsymbol{\psi}\rangle dv$, for $\boldsymbol{\psi}\in \mathcal{D}(\Omega,\mathbb{R}^3)$, where $\boldsymbol{u}$ is a piecewise smooth vector field continuous across $S$. Indeed, by Equation \eqref{distdefeq1} in the absence of disclinations, $\Curl \boldsymbol{K}^T = \boldsymbol{0}$ which, by Lemma \ref{RegularityLemma}(ii), is equivalent to the existence of a $\boldsymbol{\Omega} \in \mathcal{B}(\Omega,\mathbb{R}^3)$ such that $\boldsymbol{K} = (\nabla \boldsymbol{\Omega})^T$. Consider $\boldsymbol{W} \in \mathcal{B}(\Omega,\Skw)$ such that $\boldsymbol{\Omega}$ is the axial vector of  $\boldsymbol{W}$. Subsequently, using Equation \eqref{distdefeq2} with $\boldsymbol{A} = \boldsymbol{0}$ and $\boldsymbol{E}^Q = \boldsymbol{0}$, we obtain $\Curl (\boldsymbol{E}+\boldsymbol{W}) = \boldsymbol{0}$ which, after an application of Lemma \ref{RegularityLemma}(ii), yields the desired result. This inference can be used as a motivation for introducing the relationships between defects and kinematical quantities in the form given in Equations \eqref{distdefeq1} and \eqref{distdefeq2}.
 
 The relations  \eqref{distdefeq1} and \eqref{distdefeq2} immediately lead to their local counterpart on the interface $S$ and away from it. Using Equations \eqref{distdefeq1} and \eqref{strbend}$_2$, and Identities \ref{CurlLemma}, we obtain the local relations between the disclination densities and the bend-twist fields as
 \begin{eqnarray}
 \boldsymbol{\theta}_B=\curl \boldsymbol{\kappa}_B^T~\text{in}~\Omega-S, \label{locdisc1}
\\
\boldsymbol{\theta}_{S_1}= \left(\llbracket \boldsymbol{\kappa}_B^T \rrbracket \times \boldsymbol{n}\right)^T - \kappa \left(\boldsymbol{\kappa}_S^T \times \boldsymbol{n}\right)^T +\curl_S \boldsymbol{\kappa}_S^T~\text{on}~S,~\text{and} \label{locdisc2}
\\
\boldsymbol{\theta}_{S_2} = \left(\boldsymbol{\kappa}_S^T \times \boldsymbol{n}\right)^T~\text{on}~S. \label{locdisc3}
\end{eqnarray}
Also, using Equations \eqref{distdefeq2} and \eqref{strbend}$_1$, and Identities \ref{CurlLemma}, the dislocation densities in terms of the strain, the metric anomalies, and the bend-twist fields can be obtained as
 \begin{eqnarray}
 \boldsymbol{\alpha}_B=\curl (\boldsymbol{e}- \boldsymbol{e}^Q_B) +  \tr(\boldsymbol{\kappa}_B)\boldsymbol{I} - \boldsymbol{\kappa}_B^T~\text{in}~\Omega-S, \label{locdisl1}
\\
\boldsymbol{\alpha}_{S_1}= \left(\llbracket \boldsymbol{e}- \boldsymbol{e}^Q_B \rrbracket \times \boldsymbol{n}\right)^T + \kappa \left(\boldsymbol{e}^Q_S \times \boldsymbol{n}\right)^T -\curl_S \boldsymbol{e}^Q_S + \tr(\boldsymbol{\kappa}_S)\boldsymbol{I} - \boldsymbol{\kappa}_S^T~\text{on}~S,~\text{and}  \label{locdisl2}
\\
\boldsymbol{\alpha}_{S_2} = -\left(\boldsymbol{e}^Q_S \times \boldsymbol{n}\right)^T~\text{on}~S.  \label{locdisl3}
\end{eqnarray}
Out of the above, only Equations \eqref{locdisc1} and \eqref{locdisl1} have been previously obtained by de Wit \cite{de1981view}. The rest of the relations appear to be new. It is interesting to note that, in particular, in order to support a density of surface dislocation dipoles, it is necessary to have a non-trivial density of surface metric anomalies. These relationships provide important connections between defect densities and metric anomalies within the assumed kinematical framework given in terms of strain and bend-twist fields.

\begin{rem}
 In the absence of disclinations and metric anomalies, following the arguments given after Equation \eqref{distdefeq2}, we can infer the existence of a distribution $\boldsymbol{B}\in \mathcal{B}(\Omega,\Lin)$ such that $\boldsymbol{A}=\Curl \boldsymbol{B}$. We can write $\boldsymbol{B}(\boldsymbol{\phi})=\int_\Omega \langle \boldsymbol{\beta},\boldsymbol{\phi} \rangle dv$, for $\boldsymbol{\phi}\in \mathcal{D}(\Omega,\Lin)$, where $\boldsymbol{\beta}$ is the piecewise smooth distortion field over $\Omega-S$, possible discontinuous across $S$. Consequently, we obtain
\begin{equation}
\boldsymbol{\alpha}_B= \curl \boldsymbol{\beta}~\text{and}~
\boldsymbol{\alpha}_{S_1}=(\llbracket \boldsymbol{\beta} \rrbracket \times \boldsymbol{n})^T,
\end{equation}
in addition to $\boldsymbol{\alpha}_{S_2}=\boldsymbol{0}$. The surface dislocations $\boldsymbol{\alpha}_{S_1}$ in this form was first introduced by Bilby \cite{bilby}.
\end{rem}

\begin{rem}(Conservation laws) It follows from relations \eqref{distdefeq1} and \eqref{distdefeq2} that the distributions $\boldsymbol{A}$ and $\boldsymbol{\Theta}$ satisfy
\begin{eqnarray}
\label{ConservationLawDisclination}
\Div \boldsymbol{\Theta}^T=\boldsymbol{0}~\text{and}
\\
\label{ConservationLawDislocation}
\Div \boldsymbol{A}^T + ax( \boldsymbol{\Theta}^T -\boldsymbol{\Theta})=\boldsymbol{0}.
\end{eqnarray}
According to Theorem \ref{Poincare}, for a contractible domain $\Omega$, the above conditions are necessary and sufficient conditions for the existence of distributions $\boldsymbol{K}$ and $\boldsymbol{E}$. These conservations laws can be used to derive the local conservations laws for defect densities. We use Identities \ref{DivergenceLemma} and Equation \eqref{ConservationLawDisclination} to obtain
\begin{eqnarray}
\div \boldsymbol{\theta}_B^T = \boldsymbol{0}~\text{in} ~\Omega- S, 
\\
-\llbracket \boldsymbol{\theta}_B^T \rrbracket \boldsymbol{n} + \div_S (\boldsymbol{\theta}_{S_1}^T) + \kappa\boldsymbol{\theta}_{S_1}^T \boldsymbol{n} - \div_S (\boldsymbol{\theta}_{S_2}^T \nabla_S \boldsymbol{n}) = \boldsymbol{0}~ \text{on}~S,
\\
-\boldsymbol{\theta}_{S_1}^T \boldsymbol{n} + \div_S \boldsymbol{\theta}_{S_2}^T + \kappa\boldsymbol{\theta}_{S_2}^T \boldsymbol{n}= \boldsymbol{0} ~\text{on}~S, ~\text{and}
\\
\boldsymbol{\theta}_{S_2}^T \boldsymbol{n}= \boldsymbol{0}~\text{on}~S.
\end{eqnarray}
 Similarly, we use Identities \ref{DivergenceLemma} and Equation \eqref{ConservationLawDislocation} to obtain
\begin{eqnarray}
\div \boldsymbol{\alpha}_B^T + ax( \boldsymbol{\theta}_B^T - \boldsymbol{\theta}_B) = \boldsymbol{0}~\text{in} ~\Omega- S,
\\
-\llbracket \boldsymbol{\alpha}_B^T \rrbracket \boldsymbol{n} + \div_S (\boldsymbol{\alpha}_{S_1}^T) + \kappa\boldsymbol{\alpha}_{S_1}^T \boldsymbol{n}- \div_S ( \boldsymbol{\alpha}_{S_2}^T \nabla_S \boldsymbol{n})+ ax( \boldsymbol{\theta}_{S_1}^T-\boldsymbol{\theta}_{S_1})  =  \boldsymbol{0} ~\text{on}~S,
\\
-\boldsymbol{\alpha}_{S_1}^T \boldsymbol{n}+ \div_S (\boldsymbol{\alpha}_{S_2}^T) + \kappa\boldsymbol{\alpha}_{S_2}^T \boldsymbol{n} + ax( \boldsymbol{\theta}_{S_2}^T-\boldsymbol{\theta}_{S_2})= \boldsymbol{0} ~\text{on}~S,~\text{and}
\\
\boldsymbol{\alpha}_{S_2}^T \boldsymbol{n}= \boldsymbol{0}~ \text{on}~S.
\end{eqnarray}
\end{rem}

\begin{rem}
(Dislocation loop) We consider a form of dislocation density which is concentrated on an oriented smooth curve $L \subset \Omega$.  Assume $\boldsymbol{A} \in \mathcal{H}(\Omega,\Lin)$ such that, for $\boldsymbol{\phi}\in \mathcal{D}(\Omega,\Lin)$, we can write $\boldsymbol{A}(\boldsymbol{\phi})=\int_L \langle \boldsymbol{\alpha}_L,\boldsymbol{\phi} \rangle da$, where $\boldsymbol{\alpha}_L$ is a smooth field on $L$. Using Identity \ref{DivergenceLemma}(d), the local form of Equation \eqref{ConservationLawDislocation}, in the absence of disclinations, yields
\begin{eqnarray}
\label{dislocationloopC1}
\boldsymbol{\alpha}_L^T (\boldsymbol{I}-\boldsymbol{t}\otimes \boldsymbol{t})= \boldsymbol{0}~ \text{on}~L,
\\
\label{dislocationloopC2}
\frac{\partial}{\partial t} (\boldsymbol{\alpha}_L^T \boldsymbol{t}) = \boldsymbol{0}~\text{on}~L,
\\
\label{dislocationloopC3}
\boldsymbol{\alpha}_L^T \boldsymbol{t} =\boldsymbol{0}~\text{on}~\partial L - \partial \Omega.
\end{eqnarray} 
According to Equation \eqref{dislocationloopC1}, $\boldsymbol{\alpha}_L$ has to necessarily satisfy $\boldsymbol{\alpha}_L=\boldsymbol{t} \otimes (\boldsymbol{\alpha}_L^T \boldsymbol{t}) $, while Equation \eqref{dislocationloopC2} implies that $\boldsymbol{\alpha}_L^T \boldsymbol{t}$ is uniform along $L$. As a result, for a non-trivial dislocation density, we can infer from Equation \eqref{dislocationloopC3} that $\partial L - \partial \Omega=\emptyset$, i.e., the curve $L$ has to be either a loop or its end points should lie on the boundary of the domain. The constant vector $\boldsymbol{\alpha}_L^T \boldsymbol{t}$ should be identified with the Burgers vector associated with the dislocation loop. 
In a related work, Van Goethem \cite{van2017incompatibility}  has considered dislocation loops as tensor valued Radon measures concentrated on a closed loop and established that there exists a non square integrable strain field, absolutely continuous with respect to the volume measure, which satisfies the incompatibility condition induced by the dislocation loop. 
\end{rem}

\begin{rem} (Wall of dislocation dipoles) 
We consider a distribution $\boldsymbol{A}_h$ as introduced in Equation \eqref{ddipoles} but with $\boldsymbol{\alpha}_0$ not necessarily uniform, i.e., $\div_S (\boldsymbol{\alpha}_0^T)\neq \boldsymbol{0}$. We assume the domain to be free of disclinations and metric anomalies, as well as of dislocations in the bulk outside of the two surfaces in $\Omega -\Omega_h$. In order for the local conservation laws to be satisfied we require $\boldsymbol{\alpha}_0^T \boldsymbol{n}=\boldsymbol{0}$ in addition to a non-trivial bulk dislocation density $\hat{\boldsymbol{\alpha}}_0/h$ supported in $\Omega_h$ with the associated distribution  $\hat{\boldsymbol{A}}_h (\boldsymbol{\phi}) = \int_{\Omega_h} \langle {\hat{\boldsymbol{\alpha}}_0}/{h}, \boldsymbol{\phi} \rangle da$,
for $\boldsymbol{\phi}\in \mathcal{D}(\Omega,\Lin)$, such that the conservation law yields $- \hat{\boldsymbol{\alpha}}_0^T  \boldsymbol{n} + \div_S \boldsymbol{\alpha}_0^T =\boldsymbol{0}$. The enclosed bulk $\Omega_h$ can therefore be thought of having dislocation curves with tangents along the normal of $S$. We note that these dislocation lines remain contained inside the band and do not pierce out of either $S$ or $S_h$. For infinitesimal distance between the walls ($h\to 0$), $\boldsymbol{A}_h$ converges to a distribution corresponding to a dislocation dipole wall, as remarked earlier, and $\hat{\boldsymbol{A}}_h$ to a distribution $\hat{\boldsymbol{A}} \in \mathcal{C}(\Omega,\Lin)$ corresponding to a dislocation wall, i.e., 
$\hat{\boldsymbol{A}} (\boldsymbol{\phi}) = \int_S \langle \hat{\boldsymbol{\alpha}}_0, \boldsymbol{\phi} \rangle da$.
The derived dislocation wall has a surface density $\hat{\boldsymbol{\alpha}}_0$ such that $\hat{\boldsymbol{\alpha}}_0^T \boldsymbol{n} \neq \boldsymbol{0}$. This is in contrast with a dislocation wall which does not coincide with a dislocation dipole wall. In the latter case, considering a dislocation wall with surface density $\boldsymbol{\alpha}_S$, we necessarily require $\boldsymbol{\alpha}_S^T\boldsymbol{n}=\boldsymbol{0}$. 
\end{rem}

\subsection{Strain Incompatibility}
The bulk strain field $\boldsymbol{e}$ is compatible if and only if $\Curl \Curl \boldsymbol{E} = \boldsymbol{0}$, where 
$\boldsymbol{E}\in \mathcal{B}(\Omega,\Sym)$ is as given in Equation \eqref{strbend}$_1$. In the presence of defects and metric anomalies, the strain field is no longer compatible. We define a distribution $\boldsymbol{N} \in \mathcal{D}'(\Omega,\Sym)$ by $\boldsymbol{N}=\Curl \Curl \boldsymbol{E}$. Therefore, for $\boldsymbol{\phi}\in \mathcal{D}(\Omega,\Lin)$, 
\begin{equation}
\boldsymbol{N}(\boldsymbol{\phi})=\int_\Omega \langle \boldsymbol{\eta}_B,\boldsymbol{\phi} \rangle dv +\int_S \langle \boldsymbol{\eta}_{S_1},\boldsymbol{\phi} \rangle da+\int_S \langle \boldsymbol{\eta}_{S_2},\frac{\partial \boldsymbol{\phi}}{\partial n} \rangle da,~\text{where}
 \end{equation} 
\begin{eqnarray}
\boldsymbol{\eta}_B=\curl \curl \boldsymbol{e},
\\
\boldsymbol{\eta}_{S_1}=-\kappa \left((\llbracket \boldsymbol{e} \rrbracket\times \boldsymbol{n}\right)^T\times \boldsymbol{n})^T +\curl_S (\llbracket \boldsymbol{e} \rrbracket\times \boldsymbol{n})^T+ (\llbracket \curl \boldsymbol{e} \rrbracket \times \boldsymbol{n} )^T,~\text{and}
\\
\boldsymbol{\eta}_{S_2}=\left((\llbracket \boldsymbol{e} \rrbracket\times \boldsymbol{n})^T\times \boldsymbol{n}\right)^T
\end{eqnarray}
are incompatibility fields in the bulk, away from the interface, and on the interface. The bulk field can be identified as Kr\"oner's incompatibility tensor. We now relate these incompatibility fields to various defect and metric anomaly fields. Taking a trace of Equation \eqref{distdefeq2} and noting that $\tr(\Curl (\boldsymbol{E}-\boldsymbol{E}^Q))=0$, we obtain $\tr(\boldsymbol{A})=2\tr(\boldsymbol{K})$. Substituting this result back into Equation \eqref{distdefeq2}, and rearranging it, yields
\begin{equation}
\Curl \boldsymbol{E}= \Curl \boldsymbol{E}^Q + \boldsymbol{A} - \frac{1}{2} \tr(\boldsymbol{A}) \boldsymbol{I} +\boldsymbol{K}^T.
\end{equation}
Take another $\Curl$, and subsequently use $\boldsymbol{N} = \Curl \Curl \boldsymbol{E}$,  $\boldsymbol{\Gamma} = \boldsymbol{A} - ({1}/{2}) \tr \boldsymbol{A}$ (recall Equation \eqref{distgamma}), and Equation \eqref{distdefeq1} to obtain
\begin{equation}
\label{DefectIncompatibility}
\boldsymbol{N}=\Curl \boldsymbol{\Gamma} + \boldsymbol{\Theta} + \Curl \Curl \boldsymbol{E}^Q,
\end{equation}
The Identities  \ref{CurlLemma} can now be used to obtain the required relationships between strain incompatibilities $ \boldsymbol{\eta}_B$, $\boldsymbol{\eta}_{S_1}$, and $\boldsymbol{\eta}_{S_2}$, which are expressed in terms of strain, its derivatives, and jumps, and densities of defects and metric anomalies. We derive
 \begin{eqnarray}
 \boldsymbol{\eta}_B=\curl \boldsymbol{\gamma}_B + \boldsymbol{\theta}_B + \boldsymbol{\eta}^Q_B~\text{in}~\Omega-S, \label{etab}
\\
\label{eta1}
\boldsymbol{\eta}_{S_1}= \left(\llbracket \boldsymbol{\gamma}_B \rrbracket \times \boldsymbol{n}\right)^T - \kappa (\boldsymbol{\gamma}_{S_1} \times \boldsymbol{n})^T +\curl_S \boldsymbol{\gamma}_{S_1}  -\div_S(\nabla_S \boldsymbol{n}\times \boldsymbol{\gamma}_{S_2})+ \boldsymbol{\theta}_{S_1} + \boldsymbol{\eta}^Q_{S_1}~\text{on}~S,
\\
\label{eta2}
\boldsymbol{\eta}_{S_2} = (\boldsymbol{\gamma}_{S_1} \times \boldsymbol{n})^T - \kappa (\boldsymbol{\gamma}_{S_2} \times \boldsymbol{n})^T +\curl_S \boldsymbol{\gamma}_{S_2}+ \boldsymbol{\theta}_{S_2}+\boldsymbol{\eta}^Q_{S_2}~\text{on}~S,~\text{and}
\\
\label{eta3}
\boldsymbol{0} = (\boldsymbol{\gamma}_{S_2} \times \boldsymbol{n})^T + \boldsymbol{\eta}^Q_{S_3},
\end{eqnarray}
where $\boldsymbol{\eta}^Q_B=\curl \curl \boldsymbol{e}^Q_B$,
\begin{equation}
\begin{split}
\boldsymbol{\eta}^Q_{S_1}=  \left(\llbracket \curl \boldsymbol{e}^Q_B \rrbracket \times \boldsymbol{n} \right)^T - \kappa \left((\llbracket \boldsymbol{e}^Q_B \rrbracket\times \boldsymbol{n})^T\times \boldsymbol{n}\right)^T +\curl_S \left(\llbracket \boldsymbol{e}^Q_B \rrbracket\times \boldsymbol{n}\right)^T+ \kappa^2 \left((\boldsymbol{e}^Q_S \times \boldsymbol{n})^T \times \boldsymbol{n} \right)^T \\-\kappa \left(\curl_S \boldsymbol{e}^Q_S \times \boldsymbol{n}\right)^T  - \curl_S \left(\kappa (\boldsymbol{e}^Q_S \times \boldsymbol{n})^T\right)+ \curl_S \curl_S \boldsymbol{e}^Q_S - \div_S \left(\nabla_S \boldsymbol{n} \times (\boldsymbol{e}^Q_S \times \boldsymbol{n})^T\right),
\end{split}
\end{equation} 
\begin{equation}
\boldsymbol{\eta}^Q_{S_2}= -2\kappa \left((\boldsymbol{e}^Q_S \times \boldsymbol{n})^T \times \boldsymbol{n}\right)^T  +\left(\curl_S \boldsymbol{e}^Q_S \times \boldsymbol{n}\right)^T  + \curl_S \left(\boldsymbol{e}^Q_S \times \boldsymbol{n}\right)^T + \left((\llbracket \boldsymbol{e}^Q_B \rrbracket\times \boldsymbol{n})^T\times \boldsymbol{n}\right)^T,
\end{equation} 
and $\boldsymbol{\eta}^Q_{S_3}= \left( (\boldsymbol{e}^Q_S \times \boldsymbol{n})^T \times \boldsymbol{n}\right)^T$.
The Equations \eqref{etab}-\eqref{eta2} are the strain incompatibility equations where the left hand sides are given in terms of the strain field and the right hand sides are given in terms of the defect and the metric anomaly fields. Equation \eqref{eta3}, on the other hand, should be seen as a restriction on the nature of surface densities of dislocation dipole and metric anomaly.

\begin{rem} (Surface $S$ such that $\partial S - \partial \Omega \neq \emptyset$)
We consider a dislocation density which is concentrated on surface $S$ which has a non-trivial boundary in the interior of the body, i.e.,  $\partial S - \partial \Omega \neq \emptyset$. Accordingly, we consider a distribution $\boldsymbol{A} \in \mathcal{C}(\Omega,\Lin)$ such that, for $\boldsymbol{\phi}\in \mathcal{D}(\Omega,\Lin)$, $\boldsymbol{A}(\boldsymbol{\phi})=\int_S \langle \boldsymbol{\alpha}_S,\boldsymbol{\phi} \rangle da$. The related contortion tensor is $\boldsymbol{\gamma}_S=\boldsymbol{\alpha}_S-({1}/{2})\tr(\boldsymbol{\alpha}_S)\boldsymbol{I}$. In the absence of other defect densities and metric anomalies, the strain incompatibility relations yield $\boldsymbol{\eta}_B = \boldsymbol{0}$ in $\Omega - S$, 
 \begin{eqnarray}
\boldsymbol{\eta}_{S_1}=  \kappa (\boldsymbol{\gamma}_S \times \boldsymbol{n})^T +\curl_S \boldsymbol{\gamma}_S ~\text{on}~  S, \text{and}
\\
\boldsymbol{\eta}_{S_2} = (\boldsymbol{\gamma}_S \times \boldsymbol{n})^T ~\text{on}~  S.
\end{eqnarray}
In addition, the dislocation density must satisfy $(\boldsymbol{\gamma}_S \times \boldsymbol{\nu})^T = \boldsymbol{0}$ on $\partial S- \partial \Omega$,
where $\boldsymbol{\nu}$ is the in plane normal to $\partial S - \partial \Omega$. On the other hand, the conservation laws for dislocation density can be derived using Identity \ref{DivergenceLemma}(b) and Equation \eqref{ConservationLawDislocation} to get $\div_S \boldsymbol{\alpha}_S^T  = \boldsymbol{0}$ and $\boldsymbol{\alpha}_S^T \boldsymbol{n} =\boldsymbol{0}$ on $S$, and $\boldsymbol{\alpha}_S^T \boldsymbol{\nu}= \boldsymbol{0}$ on $\partial S- \partial \Omega$.
\end{rem}

\begin{rem} (Plane strain incompatibility conditions without metric anomalies) Assume that distributions $\boldsymbol{E}$ and $\boldsymbol{K}$ satisfy $\boldsymbol{E}\boldsymbol{e}_3 =\boldsymbol{0}$, ${\partial \boldsymbol{E}}/{\partial x_3}=\boldsymbol{0}$, and $\boldsymbol{K}=\boldsymbol{K}^P \otimes \boldsymbol{e}_3$, where $\boldsymbol{K}^P \in \mathcal{D}'(\Omega,\mathbb{R}^3)$, $\langle \boldsymbol{K}^P,\boldsymbol{e}_3 \rangle=0$, and ${\partial \boldsymbol{K}^P}/{\partial x_3}=\boldsymbol{0}$. The plane section orthogonal to $\boldsymbol{e_3}$ is denoted as $P \subset \mathbb{R}^2$. The interface $S$ is completely characterised by the planar curve $C_P=S\cap P$. Let the unit tangent to $C_P$ be $\boldsymbol{t}$. The unit normal to $C_p$ coincides with the normal $\boldsymbol{n}$ to $S$.
Under the above assumptions on $\boldsymbol{E}$ and $\boldsymbol{K}$, the distribution $\boldsymbol{A}$ corresponding to the dislocation density is necessarily of the form $\boldsymbol{A} = (\boldsymbol{A}^P \otimes \boldsymbol{e}_3)^T$, where $\boldsymbol{A}^P \in \mathcal{D}'(\Omega,\mathbb{R}^3)$ such that $\langle \boldsymbol{A}^P , \boldsymbol{e}_3 \rangle=0$ and ${\partial \boldsymbol{A}^P}/{\partial x_3}=\boldsymbol{0}$. The condition $\langle \boldsymbol{A}^P,\boldsymbol{e}_3 \rangle=0$ essentially means that only edge dislocations are admissible in the considered situation. Furthermore, the distribution $\boldsymbol{\Theta}$ corresponding to disclination density is necessarily of the form $\boldsymbol{\Theta} = {\Theta}^P \boldsymbol{e}_3 \otimes \boldsymbol{e}_3$, where ${\Theta}^P \in \mathcal{D}'(\Omega)$ and ${\partial \Theta^P}/{\partial x_3}=0$. Interestingly, for the above form of $\boldsymbol{A}$ and $\boldsymbol{\Theta}$, the conservation laws \eqref{ConservationLawDisclination} and \eqref{ConservationLawDislocation} are identically satisfied. Moreover, since $\tr \boldsymbol{A}=0$, the distribution corresponding to contortion field $\boldsymbol{\Gamma}=\boldsymbol{A}$. The incompatibility conditions, in terms of distributions, are therefore reduced to $\boldsymbol{N} = \Curl \boldsymbol{A} + \boldsymbol{\Theta}$, which for the assumed forms of  $\boldsymbol{A}$ and $\boldsymbol{\Theta}$ requires  $\boldsymbol{N}$ to be of the form $\boldsymbol{N}=N^P \boldsymbol{e}_3 \otimes \boldsymbol{e}_3$, where ${N}^P \in \mathcal{D}'(\Omega)$.
Considering dislocation and disclination densities with a bulk part and a concentration on the interface (no dipoles), the strain incompatibility relations can be written as (with obvious notation)
\begin{eqnarray}
\eta^P_B=\langle \curl \boldsymbol{\alpha}^P_B, \boldsymbol{e}_3 \rangle + \theta^P_B
~\text{in}~P-C_P, \label{planedef1}
\\
\eta^P_{S_1}=\langle \llbracket (\boldsymbol{\alpha}^P_B) \rrbracket, \boldsymbol{t}\rangle + \frac{\partial}{\partial t} \langle \boldsymbol{\alpha}^P_S,\boldsymbol{n} \rangle  +\theta^P_S ~\text{on}~ C_P,~\text{and} \label{planedef2}
\\
\eta^P_{S_2}=\langle \boldsymbol{\alpha}^P_S,\boldsymbol{t} \rangle ~\text{on}~C_P. \label{planedef3}
\end{eqnarray}
\label{planedefects}
\end{rem}

\begin{rem} (Plane strain incompatibility conditions with only interfacial metric anomalies)
We consider $\boldsymbol{E}^Q$ such that $\boldsymbol{E}^Q \boldsymbol{e}_3=\boldsymbol{0}$ and ${\partial \boldsymbol{E}^Q}/{\partial x_3}=\boldsymbol{0}$. We restrict ourselves to the case when metric anomalies are concentrated only on the surface $S$, i.e., for $\boldsymbol{\phi}\in \mathcal{D}(\Omega,\Lin)$, $\boldsymbol{E}^Q(\boldsymbol{\phi})=\int_S\langle \boldsymbol{e}^Q_S,\boldsymbol{\phi}\rangle da$. The assumed form of $\boldsymbol{E}^Q$ implies that we can express  $\boldsymbol{e}^Q_S$ as $\boldsymbol{e}^Q_S=a_1(\boldsymbol{t}\otimes \boldsymbol{t})+a_2(\boldsymbol{t}\otimes \boldsymbol{n}+\boldsymbol{n}\otimes \boldsymbol{t})+a_3(\boldsymbol{n}\otimes \boldsymbol{n})$, where $a_1$, $a_2$, and $a_3$ depend only the parameter $t$ on $C_P$. As in the preceding remark, $\boldsymbol{N}=N^P \boldsymbol{e}_3 \otimes \boldsymbol{e}_3$, where ${N}^P \in \mathcal{D}'(\Omega)$. The condition $( (\boldsymbol{e}^Q_S \times \boldsymbol{n})^T \times \boldsymbol{n})^T = \boldsymbol{0}$ implies that $a_1=0$. The nontrivial strain compatibility equations in the present case are 
\begin{eqnarray}
\eta^P_{S_1}=a_3''+2(ka_2)' ~\text{on} ~C_p~\text{and} \label{planemet1}
\\
\eta^P_{S_2}=2a_2'-ka_3 ~\text{on} ~C_p,  \label{planemet2}
\end{eqnarray}
where the superposed prime denotes the derivative with respect to $t$.
\label{planemetric}
\end{rem}

\subsection{Nilpotent Defect Densities}
It is clear from the strain incompatibility relations \eqref{etab}-\eqref{eta2} that it is possible to have non-trivial defect and metric anomaly densities such that they would not contribute to incompatibility, i.e., when the right hand sides of these relations are identically zero. Such defect densities, termed nilpotent, exist without acting as a source for internal stresses in the body. 
In the absence of metric anomalies, the distributions associated with nilpotent dislocations and disclinations will satisfy
\begin{equation}
\Curl \boldsymbol{\Gamma} + \boldsymbol{\Theta}=\boldsymbol{0}.
\end{equation} 
When dislocations are also absent then there can be no nontrivial nilpotent disclination density. On the other hand, when disclinations are absent then nilpotent dislocation densities satisfy $\Curl \boldsymbol{\Gamma} =\boldsymbol{0}$ which, by Theorem \ref{PoincareOneForm}, implies that $\boldsymbol{\Gamma}$ must be expressible as a gradient of a vector valued distribution. If we consider only a surface density of dislocations, i.e., $\boldsymbol{\alpha}_{S_1}$, and neglect others, then the  nilpotent dislocation density represents a grain boundary $S$ where $\curl_S \boldsymbol{\gamma}_{S_1}=\boldsymbol{0}$ and $\boldsymbol{\gamma}_{S_1} \times \boldsymbol{n}=\boldsymbol{0}$.

Nilpotent dislocations in the case of plane deformation, as discussed in Remark \ref{planedefects}, and without disclinations correspond to $\Curl \boldsymbol{A}^P = \boldsymbol{0}$. Theorem \ref{PoincareOneForm} then implies that there exists a scalar valued distribution $R \in \mathcal{D}'(\Omega)$ such that $\boldsymbol{A}^P = \nabla R$. If we consider only a bulk and a surface dislocation density (and ignore surface dipoles) then this form of $\boldsymbol{A}^P$ implies that $R$ is a piecewise smooth function discontinuous across the curve $C_P$; the field $R$ can be interpreted as the orientation of the lattice at each point. The condition \eqref{planedef3} with $\eta^P_{S_2}=0$ implies that $\boldsymbol{\alpha}^P_S$ at each point on the curve $C_P$ is along the normal to $C_P$, i.e., $\boldsymbol{\alpha}^P_S = |\boldsymbol{\alpha}^P_S|\boldsymbol{n}$. Here, $|\boldsymbol{\alpha}^P_S|$ is the jump in $R$ across $C_P$ or, in other words, the misorientation across the interface. On the other hand, the condition \eqref{planedef2}, with $\eta^P_{S_2}=0$ and no disclinations, reduces to
\begin{equation}
\langle \llbracket \boldsymbol{\alpha}^P_B \rrbracket, \boldsymbol{t} \rangle + \frac{\partial}{\partial t}\langle \boldsymbol{\alpha}^P_S,\boldsymbol{n} \rangle=\boldsymbol{0}.
\end{equation}
The above equation implies that, whenever the bulk dislocation density is continuous across $C_P$, $|\boldsymbol{\alpha}^P_S|$ is constant along $C_P$. We then have a grain boundary with constant misorientation at each point of the boundary. A grain boundary with variable misorientation along the boundary can exist only if we have a non-trivial jump in the bulk dislocation density across the boundary.

Finally, we assume all the defect densities to be absent and consider only a surface density of metric anomalies over $S$, i.e., we take only $\boldsymbol{e}^Q_S$ to be non-zero. We investigate the implications of requiring such a metric anomaly field to be nilpotent. The distribution $\boldsymbol{E}^Q_S$, defined in \eqref{distmetano}, with only $\boldsymbol{e}^Q_S$ present has to satisfy $\Curl \Curl \boldsymbol{E}^Q_S = \boldsymbol{0}$. One consequence of this relation is $((\boldsymbol{e}^Q_S \times \boldsymbol{n})^T \times \boldsymbol{n})^T=\boldsymbol{0}$ which implies that $\boldsymbol{e}^Q_S= ({1}/{2})(\boldsymbol{g}\otimes \boldsymbol{n}+\boldsymbol{n}\otimes \boldsymbol{g})$, where $\boldsymbol{g} \in C^{\infty}(S,\mathbb{R}^3)$. The nilpotence of $\boldsymbol{E}^Q$ is then equivalent to the existence of $\boldsymbol{U} \in \mathcal{B}(\Omega,\mathbb{R}^3)$ with a piecewise smooth bulk density $\boldsymbol{u}$ whose jump at $S$ is equal to $-\boldsymbol{g}$ and which satisfies $(1/2) (\nabla \boldsymbol{u} + (\nabla \boldsymbol{u})^T) = \boldsymbol{0}$ in $\Omega-S$. Alternatively, we can consider $\boldsymbol{u}$ to be non-trivial only in a domain $\Omega^+$, on one side of $S$, and zero in rest of the domain. On the boundary of $\Omega^+$ which coincides with $S$, $\boldsymbol{u} = \boldsymbol{g}$. Therefore if we consider a domain $\Omega^+$, with $S$ as the boundary where a displacement boundary condition is specified as $\boldsymbol{u}=\boldsymbol{g}$, the nilpotence of $\boldsymbol{E}^Q$ is equivalent to whether the displacement boundary condition in consistent with the rotation and translation of domain $\Omega^+$.

For the planar case, as discussed in Remark \ref{planemetric}, if we additionally assume that the quasi plastic strain is a result of only a slip across the boundary, i.e., $a_3 = \langle \boldsymbol{e}^Q_S, \boldsymbol{n}\otimes \boldsymbol{n}\rangle=0$. It then follows immediately from Equations  \eqref{planemet1} and \eqref{planemet2} that a non-trivial $\boldsymbol{E}^Q$, with only surface density, can be nilpotent only if $k'=0$, i.e., when the curve $C_P$ is linear or circular and if the slip is uniform, i.e.,  $a_2 = \langle \boldsymbol{e}^Q_S, \boldsymbol{t}\otimes \boldsymbol{n}\rangle$ is constant along $C_P$. For a linear interface this corresponds to translation of $\Omega^+$, with $\Omega^-$ fixed, and for a circular interface this corresponds to a rotation of $\Omega^+$, with $\Omega^-$ fixed. For an interface with non-uniform curvature, a quasi plastic strain with non-trivial slip can not be nilpotent; the non uniformity of curvature will always act as a source of strain incompatibility. 

\section{Conclusion}
\label{conc}
We have used the theory of distributions to discuss the problems of both strain compatibility and strain incompatibility, the latter arising as a result of inhomogeneities in the form of defects and metric anomalies. The main focus of our work has been to develop a framework which incorporates strain and inhomogeneity fields less regular than previously discussed in the literature. In particular, we have allowed the bulk fields to be piecewise smooth, possibly discontinuous over a singular interface, and also for smooth fields concentrated on the interface. Our work is amenable for also including concentrations over curves and points. The overall framework can be possibly extended to further relax the regularity of various fields. Our work, it seems, can be directly related to the theory of currents \cite{de2012differentiable}, which can provide a natural setting for problems in mechanics with less regularity. Some preliminary attempts in using theory of currents to model singular defects in solids can be found in the recent work of Epstein and Segev \cite{epstein2014geometric}. One lacuna that we find in our work is to provide physical interpretations to the distributions that we have constructed out of strains and inhomogeneity fields. Such interpretations would lead us to apply the framework to more sophisticated problems, for instance those afforded by nonlinear strain fields. One possible way towards this end would be to understand the distributions, in their own right, within an appropriate differential geometric setup.

\appendix
\appendixpage

\section{Proof of Identities in Section \ref{ui}}
\label{appid}

\subsection{Proof of Identities \ref{GradientLemma}}

(a) For $B \in \mathcal{B}(\Omega)$ and $\boldsymbol{\psi} \in \mathcal{D}(\Omega,\mathbb{R}^3)$, 
$\nabla B (\boldsymbol{\psi})= - B(\div \boldsymbol{\psi}) =\int_\Omega \langle \nabla b , \boldsymbol{\psi}\rangle dx - \int_S \langle \llbracket b \rrbracket \boldsymbol{n}, \boldsymbol{\psi} \rangle da$.

\noindent (b) For $C \in \mathcal{C}(\Omega)$ and $\boldsymbol{\psi} \in \mathcal{D}(\Omega,\mathbb{R}^3)$, let $\overline{c} \in C^\infty(\Omega)$ be a smooth extension of ${c} \in C^\infty(S)$
so as to write $\nabla C (\boldsymbol{\psi})= - C(\div \boldsymbol{\psi}) = -\int_S c (\div \boldsymbol{\psi}) da  = -\int_S (\div (\overline{c}\boldsymbol{\psi}) - \langle \nabla \overline{c}, \boldsymbol{\psi} \rangle) da$. Subsequently, use $\div (\overline{c}\boldsymbol{\psi}) = \div_S (c\boldsymbol{\psi}) + \langle \nabla(\overline{c}\boldsymbol{\psi})\boldsymbol{n},\boldsymbol{n} \rangle$, $\langle \nabla \overline{c}, \boldsymbol{\psi} \rangle = \langle \nabla_S c,\boldsymbol{\psi} \rangle + \langle \nabla \overline{c} , \boldsymbol{n} \rangle \langle \boldsymbol{\psi},\boldsymbol{n} \rangle$ on $S$, and the divergence theorem to get the desired result.

\noindent (c) For $F \in \mathcal{F}(\Omega)$ and $\boldsymbol{\psi} \in \mathcal{D}(\Omega,\mathbb{R}^3)$,
$\nabla F (\boldsymbol{\psi})= - F(\div \boldsymbol{\psi}) = -\int_S f {\partial (\div \boldsymbol{\psi})}/{\partial n} da$. But  $ {\partial (\div \boldsymbol{\psi})}/{\partial n} = \langle \nabla(\div \boldsymbol{\psi}) , \boldsymbol{n}\rangle =\langle\div_S (\nabla \boldsymbol{\psi})^T,\boldsymbol{n}\rangle + \langle (\nabla(\nabla \boldsymbol{\psi}) ) \boldsymbol{n}\otimes \boldsymbol{n}, \boldsymbol{n}\rangle$, on one hand, and 
$\langle\div_S \left( (\nabla \boldsymbol{\psi})^T \right),\boldsymbol{n}\rangle =\div_S ({\partial \boldsymbol{\psi}}/{\partial n}) -  \langle \nabla_S \boldsymbol{n}, \nabla \boldsymbol{\psi} \rangle$, on the other. Upon substitution, and using the chain rule for derivatives, we can obtain $\nabla F (\boldsymbol{\psi}) =$
\begin{equation}
 -\int_S  \left(\div_S \left(f \frac{\partial \boldsymbol{\psi}}{\partial n} \right)-\left\langle \nabla_S f, \frac{\partial \boldsymbol{\psi}}{\partial n} \right\rangle- \div_S(f(\nabla_S \boldsymbol{n}) \boldsymbol{\psi})+ \langle \div_S (f\nabla_S \boldsymbol{n}),  \boldsymbol{\psi} \rangle + \langle(\nabla(\nabla \boldsymbol{\psi})) \boldsymbol{n}\otimes \boldsymbol{n}, \boldsymbol{n}\rangle \right) da, \nonumber
\end{equation}
which immediately yields the result.

\noindent (d) For $H \in \mathcal{H}(\Omega)$ and $\boldsymbol{\psi} \in \mathcal{D}(\Omega,\mathbb{R}^3)$, we have $\nabla H (\boldsymbol{\psi})= - H(\div \boldsymbol{\psi}) = -\int_L h (\div \boldsymbol{\psi}) dl = -\int_L (h \langle \nabla \boldsymbol{\psi},(\boldsymbol{I}-\boldsymbol{t}\otimes \boldsymbol{t})\rangle + \langle h \boldsymbol{t}, {\partial \boldsymbol{\psi}}/{\partial t}\rangle) dl$, leading to the desired identity.

\subsection{Proof of Identities \ref{DivergenceLemma}}

(a) For $\boldsymbol{B} \in \mathcal{B}(\Omega,\mathbb{R}^3)$ and ${\psi} \in \mathcal{D}(\Omega)$, $\Div \boldsymbol{B} ({\psi})=-\boldsymbol{B}(\nabla {\psi}) =  -\int_\Omega \langle \boldsymbol{b},\nabla {\psi} \rangle dv$, which on using the divergence theorem yields the result.

\noindent (b) For $\boldsymbol{C} \in \mathcal{C}(\Omega,\mathbb{R}^3)$ and ${\psi} \in \mathcal{D}(\Omega)$,
$\Div \boldsymbol{C} ({\psi})=-\boldsymbol{C}(\nabla {\psi})=- \int_S \langle \boldsymbol{c}, \nabla {\psi} \rangle da =-\int_S \div_S (\boldsymbol{c}  {\psi}) da + \int_S (\div_S \boldsymbol{c} ){\psi}  da - \int_S \left\langle \boldsymbol{c}, \boldsymbol{n} \right\rangle ({\partial {\psi}}/{\partial {n}}) da$. The desired identity follows upon using the divergence theorem.

\noindent (c) For $\boldsymbol{F} \in \mathcal{F}(\Omega,\mathbb{R}^3)$ and ${\psi} \in \mathcal{D}(\Omega)$,
$\Div \boldsymbol{F}({\psi}) = -\boldsymbol{F}(\nabla {\psi}) = -\int_S  \langle \boldsymbol{f},\nabla(\nabla {\psi}) \boldsymbol{n}  \rangle da$.
Using $\nabla(\nabla {\psi}) \boldsymbol{n} = (\boldsymbol{I}-\boldsymbol{n}\otimes \boldsymbol{n})(\nabla(\nabla {\psi}) \boldsymbol{n}) +(\boldsymbol{n}\otimes \boldsymbol{n})(\nabla(\nabla {\psi}) \boldsymbol{n})$
and $(\boldsymbol{I}-\boldsymbol{n}\otimes \boldsymbol{n})(\nabla(\nabla {\psi}) \boldsymbol{n})  = \nabla_S ({\partial \Psi}/{\partial n}) -   \nabla_S \boldsymbol{n} \nabla {\psi}$ we get
\begin{equation}
\Div \boldsymbol{F}({\psi}) = -\int_S  \left\langle \boldsymbol{f},\left(\nabla_S \left(\frac{\partial \psi}{\partial n} \right) -   \nabla_S \boldsymbol{n}\nabla {\psi}\right)  \right\rangle da - \int_S \langle \boldsymbol{f},\boldsymbol{n} \rangle \langle \nabla(\nabla {\psi}), \boldsymbol{n} \otimes \boldsymbol{n} \rangle da, \nonumber
\end{equation}
which after some manipulation produces the required identity. 

\noindent (d) For $\boldsymbol{H} \in \mathcal{H}(\Omega,\mathbb{R}^3)$ and ${\psi} \in \mathcal{D}(\Omega)$, we have
$\Div \boldsymbol{H} ({\psi})= -\boldsymbol{H}(\nabla {\psi})=-\int_L \langle \boldsymbol{h},\nabla {\psi} \rangle dl  = - \int_L \langle \boldsymbol{h},(\boldsymbol{I}-\boldsymbol{t} \otimes \boldsymbol{t}) \nabla{\psi} \rangle dl -\int_L \langle \boldsymbol{h},({\partial {\psi}}/{\partial t})\boldsymbol{t} \rangle dl$. The final identity is immediate.

\subsection{Proof of Identities \ref{CurlLemma}}
(a) For $\boldsymbol{B} \in \mathcal{B}(\Omega,\mathbb{R}^3)$ and $\boldsymbol{\phi} \in \mathcal{D}(\Omega,\mathbb{R}^3)$,
$\Curl \boldsymbol{B} (\boldsymbol{\phi})= \boldsymbol{B} (\curl \boldsymbol{\phi})= \int_{\Omega} \langle \boldsymbol{b},\curl \boldsymbol{\phi}  \rangle dv = \int_{\Omega} (\div(\boldsymbol{\phi}\times \boldsymbol{b})+\langle \curl \boldsymbol{b},\boldsymbol{\phi} \rangle) dv$. The result follows after using the divergence theorem.

\noindent (b) For $\boldsymbol{C} \in \mathcal{C}(\Omega,\mathbb{R}^3)$ and $\boldsymbol{\phi} \in \mathcal{D}(\Omega,\mathbb{R}^3)$, we have $\Curl \boldsymbol{C} (\boldsymbol{\phi})= \boldsymbol{C} (\curl \boldsymbol{\phi})= \int_{S} \langle \boldsymbol{c},\curl \boldsymbol{\phi}  \rangle da = \int_{S} \langle \boldsymbol{c},\curl_S \boldsymbol{\phi}-({\partial \boldsymbol{\phi}}/{\partial n}) \times \boldsymbol{n}  \rangle da$.
Recall the identity $\div_S (\boldsymbol{u}\times \boldsymbol{v})=\langle \curl_S \boldsymbol{u},\boldsymbol{v}\rangle-\langle \boldsymbol{u},\curl_S \boldsymbol{v} \rangle$, for $\boldsymbol{u},\boldsymbol{v} \in C^{\infty}(S,\mathbb{R}^3)$, to get 
\begin{equation}
\Curl \boldsymbol{C} (\boldsymbol{\phi})=  \int_{S} \div_S (\boldsymbol{\phi}\times \boldsymbol{c}) da +\int_S \langle \boldsymbol{\phi},\curl_S \boldsymbol{c}\rangle da-\int_S \left\langle \boldsymbol{c} ,\frac{\partial \boldsymbol{\phi}}{\partial n} \times \boldsymbol{n}  \right\rangle da, \nonumber
\end{equation}
which immediately lead to the pertinent identity.

\noindent (c) For $\boldsymbol{F} \in \mathcal{F}(\Omega,\mathbb{R}^3)$ and $\boldsymbol{\phi} \in \mathcal{D}(\Omega,\mathbb{R}^3)$, $\Curl \boldsymbol{F} (\boldsymbol{\phi}) = \boldsymbol{F}(\curl \boldsymbol{\phi} ) = \int_S  \langle \boldsymbol{f},{\partial (\curl \boldsymbol{\phi})}/{\partial n}  \rangle da$. Use the skew part of the identity 
$\nabla_S ({\partial \boldsymbol{\phi}}/{\partial n})=\nabla(\nabla \boldsymbol{\phi})\boldsymbol{n}-(\nabla(\nabla \boldsymbol{\phi})\boldsymbol{n}\otimes\boldsymbol{n})\otimes\boldsymbol{n} +\nabla \boldsymbol{\phi}\nabla_S \boldsymbol{n}$ to obtain 
$\curl_S({\partial \boldsymbol{\phi}}/{\partial n})={\partial (\curl \boldsymbol{\phi})}/{\partial n}+(\nabla(\nabla \boldsymbol{\phi})\boldsymbol{n}\otimes\boldsymbol{n})\times\boldsymbol{n} + ax(\nabla \boldsymbol{\phi}\nabla_S \boldsymbol{n}-(\nabla \boldsymbol{\phi}\nabla_S \boldsymbol{n})^T)$.
Furthermore, we note that
\begin{equation}
\int_S \left\langle \boldsymbol{f}, \curl_S \left(\frac{\partial \boldsymbol{\phi}}{\partial n}\right) \right\rangle da= \int_S   \left\langle -\kappa \left(\boldsymbol{f}\times \boldsymbol{n}\right) +\curl_S\left(\boldsymbol{f}\right),\frac{\partial \boldsymbol{\phi}}{\partial n} \right\rangle  da\\+\int_{\partial S - \partial \Omega} \left\langle\left( \boldsymbol{f}\times \boldsymbol{\nu}\right), \frac{\partial \boldsymbol{\phi}}{\partial n} \right\rangle dl, \nonumber
\end{equation}
\begin{equation}
\int_S \langle \boldsymbol{f}, \left(\nabla(\nabla \boldsymbol{\phi})\boldsymbol{n}\otimes\boldsymbol{n}\right)\times\boldsymbol{n} \rangle da=-\int_S \langle \boldsymbol{f}\times\boldsymbol{n}, \left(\nabla(\nabla \boldsymbol{\phi})\boldsymbol{n}\otimes\boldsymbol{n}\right) \rangle da, \nonumber
\end{equation}
and $\langle \boldsymbol{f} ,ax(\nabla \boldsymbol{\phi}\nabla_S \boldsymbol{n}-(\nabla \boldsymbol{\phi}\nabla_S \boldsymbol{n})^T)  \rangle = \langle \boldsymbol{\tilde{f}},\nabla \boldsymbol{\phi}\nabla_S \boldsymbol{n} \rangle=  -\langle (\nabla_S \boldsymbol{n} \times \boldsymbol{f})^T,\nabla_S   \boldsymbol{\phi} \rangle = \langle \div_S (\nabla_S \boldsymbol{n} \times \boldsymbol{f})^T,\boldsymbol{\phi} \rangle - \div_S ((\nabla_S \boldsymbol{n} \times \boldsymbol{f})   \boldsymbol{\phi})$, where $\boldsymbol{\tilde{f}}$ is the skew symmetric tensor whose axial vector is $\boldsymbol{f}$. Consequently, $\int_S \langle \boldsymbol{f} ,ax(\nabla \boldsymbol{\phi}\nabla_S \boldsymbol{n}-(\nabla \boldsymbol{\phi}\nabla_S \boldsymbol{n})^T)  \rangle da =$
\begin{equation}
 \int_S \left\langle \div_S \left(\nabla_S \boldsymbol{n} \times \boldsymbol{f}\right)^T,\boldsymbol{\phi} \right\rangle da -\int_{\partial S - \partial \Omega} \left\langle (\nabla_S \boldsymbol{n} \times \boldsymbol{f})\boldsymbol{\phi},\boldsymbol{\nu} \right\rangle dl +\int_{S} \kappa \left\langle (\nabla_S \boldsymbol{n}\times \boldsymbol{f})\boldsymbol{\phi},\boldsymbol{n} \right\rangle da. \nonumber
\end{equation}
The desired identity follows after combining the above results. 

\noindent(d) For $\boldsymbol{H} \in \mathcal{H}(\Omega,\mathbb{R}^3)$ and $\boldsymbol{\phi} \in \mathcal{D}(\Omega,\mathbb{R}^3)$,
$\Curl \boldsymbol{H} (\boldsymbol{\phi})=\boldsymbol{H}(\curl \boldsymbol{\phi})= \int_L \langle \boldsymbol{h} , \curl \boldsymbol{\phi})  \rangle dl= \int_L \langle \boldsymbol{h} , \curl_{t} \boldsymbol{\phi}  \rangle dl -\int_L \langle \boldsymbol{h} , ({\partial \boldsymbol{\phi}}/{\partial t} \times \boldsymbol{t})  \rangle dl$. The required result is imminent.

\section{A Lemma for Theorem \ref{Poincare}}
\label{poinapp}

A distribution $T\in \mathcal{D}'(\Omega)$ is said to be of order $m$ if, for any compact set $K \subset \Omega$, there exists a finite $M \in \mathbb{R}$ such that, for any smooth function $\phi$ supported in $K$, 
$|T(\phi)| \leq M \Sigma_{|\alpha| \leq m} |\text{sup} (\partial^\alpha \phi)|$, where $\partial^\alpha$ denotes the $\alpha$ order derivative of $\phi$. In particular,
$T$ is of order $0$ if 
$|T(\phi)| \leq M  |\text{sup} (\phi)|$. 
\begin{lemma}
\label{CohomologoustoSmoothForm}
For a $\boldsymbol{T} \in \mathcal{D}'(\Omega,\mathbb{R}^3)$, which satisfies $\Div \boldsymbol{T}=0$, there exists $\boldsymbol{u} \in C^{\infty} (\Omega, \mathbb{R}^3)$ and  $\boldsymbol{S} \in \mathcal{D}'(\Omega,\mathbb{R}^3)$ such that 
\begin{equation}
 \boldsymbol{T_u} - \boldsymbol{T} = \Curl \boldsymbol{S},
\end{equation}
where $\boldsymbol{T_u} \in \mathcal{D}'(\Omega,\mathbb{R}^3)$ is given by $\boldsymbol{T_u} (\boldsymbol{\phi})=\int_\Omega \langle \boldsymbol{u}, \boldsymbol{\phi} \rangle dv$ for all $\boldsymbol{\phi} \in \mathcal{D}(\Omega,\mathbb{R}^3)$.
\begin{proof}
Consider a map $H^y : [0,1]\times \mathbb{R}^3 \to \mathbb{R}^3$ given by
$H^y (t,x)=x + t \psi (x) y$,
where $\psi$ is a smooth scalar field over $\mathbb{R}^3$ such that $\psi(x) = 0$ for ${x} \notin \Omega$ but $0<\psi \leq 1$, $|\nabla \psi | \leq 1$ whenever $x \in \Omega$, and ${y} \in \mathbb{R}^3$ is such that $|{y}| <1$.
It can be shown that, for any $t \in [0,1]$, $H^y : [0,1]\times \Omega \to \Omega$. For $\boldsymbol{\phi} \in \mathcal{D}(\Omega,\mathbb{R}^3)$, we introduce
\begin{equation}
\boldsymbol{S}^y (\boldsymbol{\phi}) = \int_{0}^{1} \langle \boldsymbol{T} , (\boldsymbol{\phi} (H^y(t,\boldsymbol{x}) \times y)\psi (x)\rangle dt.
\end{equation}
To check that $\boldsymbol{S}^y \in \mathcal{D}'(\Omega,\mathbb{R}^3)$ it is sufficient to note that ${S}^y_i$ defines a linear functional on $\mathcal{D}(\Omega)$ and that a sequence of smooth functions $\phi_m$ converging to 0 implies the convergence of $(\boldsymbol{\phi} (H^y(t,x)) \times y)_i\psi (x)$, and consequently of ${S}^y_i(\phi_m)$, to $0$. Moreover, for $\boldsymbol{\phi} \in \mathcal{D}(\Omega,\mathbb{R}^3)$, $\Curl \boldsymbol{S}^y (\boldsymbol{\phi}) = \boldsymbol{S}^y (\curl \boldsymbol{\phi})=$
\begin{equation}
 \int_{0}^{1} \left\langle {T}_i, \left(\frac{\partial}{\partial t} \left( \phi_i(H^y (t,x))+ \phi_j (H^y (t,x)) y_j t \frac{\partial \psi}{\partial x_i} \right) - \frac{\partial}{\partial x_i} (\phi_j (H^y (t,x) y_j \psi) \right)  \right\rangle dt, \nonumber
\end{equation}
which, on using $\Div \boldsymbol{T}=0$ and $H^y (0,x)=x$, yields
\begin{equation}
\Curl \boldsymbol{S}^y (\boldsymbol{\phi}) = \left\langle {T}_i,  \left(\phi_i({x}+\psi({x}) y)+ \phi_j({x}+\psi({x}) y) y_j  \frac{\partial \psi}{\partial x_i}\right)\right\rangle - \boldsymbol{T}(\boldsymbol{\phi}).
\end{equation}
Let $\rho \in C^{\infty}( \mathbb{R}^3 )$ be a smooth function supported over a ball of unit radius, centred at the origin, such that it depends only on $|\boldsymbol{x}|$ and satisfies $\int_{\mathbb{R}^3} \rho({x}) dv =1$. Given $\epsilon >0$, the function $\rho_\epsilon = \epsilon^{-3} \rho(x/{\epsilon})$ is supported in a ball of radius $\epsilon$ such that $\int_{\mathbb{R}^3} \rho_\epsilon({x}) dv =1$. For $\boldsymbol{S} \in \mathcal{D}'(\Omega,\mathbb{R}^3)$, defined as $\boldsymbol{S}=\int_{B(0,\epsilon)} \boldsymbol{S}^y \rho_\epsilon(y) dv_y$, where $B(0,\epsilon)$ is a ball of radius $\epsilon$ centred at the origin,  $\Curl \boldsymbol{S}(\boldsymbol{\phi})=\int_{B(0,\epsilon)} \Curl \boldsymbol{S}^y (\boldsymbol{\phi})\rho_\epsilon(y) dv_y  =$
 \begin{equation}
\int_{B(0,\epsilon)} \left( \langle {T}_i,  (\phi_i({x}+\psi({x}) y)+ \phi_j({x}+\psi({x}) y) y_j  \frac{\partial \psi}{\partial x_i})\rangle \rho_\epsilon (y) \right) dv_y - \boldsymbol{T}(\boldsymbol{\phi}). \nonumber
 \end{equation}
 We can henceforth write $\Curl \boldsymbol{S}=\boldsymbol{T}_1-\boldsymbol{T}$, where $ \boldsymbol{T}_1 (\boldsymbol{\phi})= \boldsymbol{T} (\boldsymbol{\phi}^\epsilon)$, 
 \begin{equation}
 \phi^\epsilon_i(x) =  \int  \rho_\epsilon \left(\frac{z-x}{\psi(x)}\right) \frac{\phi_i(z)}{\psi(x)} dv_z + \int  \left(\rho_\epsilon \left(\frac{z-x}{\psi(x)}\right) \frac{z_j - x_j}{\psi (x)}\frac{\partial \psi}{\partial x_i}\right) \frac{\phi_j(z)}{\psi(x)} dv_z, \nonumber
 \end{equation}
 and $z = x+\psi(x) y$.
 Since $\rho_\epsilon $ is smooth, its derivatives remain bounded and the supremum norm of $\boldsymbol{\phi}^\epsilon$ and all the partial derivatives of $\boldsymbol{\phi}^\epsilon$ are controlled by the supremum norm of $|\boldsymbol{\phi}|$. Therefore, there exist a $\boldsymbol{u} \in C^{\infty} (\Omega, \mathbb{R}^3)$ such that $\boldsymbol{T}_1=\boldsymbol{T_u}$ leading us to our assertion.
\end{proof}
\end{lemma}
 
\bibliography{incomp}
\bibliographystyle{plain}

\end{document}